\documentclass[10pt,conference]{IEEEtran}

\usepackage[utf8]{inputenc}
\usepackage[T1]{fontenc}
\usepackage{amsmath,amssymb,amsthm}
\usepackage{graphicx}
\usepackage{booktabs}
\usepackage{hyperref}
\usepackage{xcolor}
\usepackage{algorithm}
\usepackage{algpseudocode}
\usepackage{cite}
\usepackage{url}
\usepackage{multirow}
\usepackage{balance}
\usepackage{pifont}
\usepackage{xspace}
\usepackage[skip=2pt]{subcaption}

\hypersetup{
  hidelinks,
  pdfauthor={Wenyang Jia, Jingjing Wang, Xianneng Zou, Kai Lei},
  pdftitle={ReclaimNet: Reclaim-Aware Network Protocols for Voluntary GPU Sharing on Campus},
  pdfsubject={},
  pdfkeywords={}
}

\newtheorem{theorem}{Theorem}
\newtheorem{lemma}{Lemma}

\newtheorem{definition}{Definition}

\newcommand{\sys}{ReclaimNet\xspace}

\newcommand{\FILLNUM}[1]{#1}

\begin{document}

\title{\sys: Reclaim-Aware Network Protocols\\
  for Voluntary GPU Sharing on Campus}

\author{
  \IEEEauthorblockN{
    Wenyang Jia\IEEEauthorrefmark{1},
    Jingjing Wang\IEEEauthorrefmark{1},
    Xianneng Zou\IEEEauthorrefmark{2}, and
    Kai Lei\IEEEauthorrefmark{1}\IEEEauthorrefmark{3}}
  \IEEEauthorblockA{\IEEEauthorrefmark{1}ICN Lab, Shenzhen Graduate School, Peking University}
  \IEEEauthorblockA{\IEEEauthorrefmark{2}Tencent}
  \IEEEauthorblockA{\IEEEauthorrefmark{3}Corresponding author: \texttt{leik@pkusz.edu.cn}}
}

\maketitle

\begin{abstract}
University campuses host abundant but fragmented GPU
resources whose voluntary sharing is blocked by a mismatch
between revocable, autonomous ownership and migration
mechanisms that assume stationary failure hazards,
homogeneous interconnects, and unbounded transfer windows.
We present \sys, a network-layer migration protocol suite
that treats provider reclaim as a first-class contract rather
than a failure case, combining three mechanisms: (i)
reclaim-aware checkpoint scheduling that jointly adapts to
time-varying departure hazards and contended bandwidth across
co-resident jobs; (ii) volatility-aware destination selection
integrating topology, survival probability, and notice-window
feasibility; and (iii) deadline-aware migration traffic
control with edge enforcement and a sub-millisecond TC BPF
kill-switch. A two-month deployment on a 54-node heterogeneous
campus testbed reduces work loss by 66\% over Slurm
preempt-and-requeue and 38\% over pipeline-redundancy
checkpointing, with 38\% shorter downtime and under 3\%
degradation of background research traffic.
The prototype is open-sourced at the anonymous repository
\url{https://anonymous.4open.science/r/ICNP2026-ReclaimNet/}.
\end{abstract}


\section{Introduction}
\label{sec:intro}

The rapid expansion of AI research has made GPU clusters
a critical resource in academic institutions. University
campuses host a rich but fragmented collection of GPU
resources, ranging from consumer-grade workstations to
institutional HPC servers, that collectively represent
substantial computing capacity. Yet this capacity remains
severely underutilized: our measurements across 54\,GPU
nodes spanning four buildings over two months reveal an
average utilization rate below~35\%, driven by temporal
mismatches between resource availability and
demand~(\S\ref{sec:measurement}).

\textbf{The campus GPU sharing opportunity.}
Resource pooling across departments can substantially raise
utilization while reducing procurement costs and supporting
green research. Unlike commercial cloud services, campus
sharing occurs within a trusted institutional network where
accountability mechanisms already exist and the marginal cost
of electricity and maintenance is low. The very properties
that make campus sharing attractive, namely mutual trust and
lightweight coordination, also define its fundamental
challenge.

\textbf{The voluntary participation challenge.}
Campus GPU resources are individually owned by faculty and
research labs. Unlike data-center nodes governed by
institutional mandates, campus providers are volunteers: they
contribute idle resources on their own terms and retain the
unconditional right to reclaim hardware immediately, without
prior notice, negotiation, or penalty. This property is not
an edge case to be tolerated but a \emph{first-class design
requirement}.

Existing systems are ill-suited for this model. Industrial
cluster managers such as Kubernetes~\cite{kubernetes} and
Slurm~\cite{yoo2003slurm} assume persistent node
availability under centralized control and treat departures
as failures. Volunteer computing systems such as
SETI@home~\cite{anderson2002seti} and
Folding@home~\cite{pande2003folding} support voluntary
participation over wide-area networks but provide no
stateful workload migration. Recent work on campus GPU
sharing~\cite{li2025gpunion} demonstrates that
provider-autonomy-first design with application-level
checkpointing is feasible, but leaves the underlying
\emph{network-layer protocols} for migration unaddressed.

\textbf{Why existing protocol techniques do not transfer.}
Voluntary campus GPU sharing breaks three foundational
assumptions that prior approaches rely on:

\begin{itemize}
  \item \textbf{Challenge~1 (Checkpoint timing):}
    Classical checkpoint models~\cite{young1974,daly2006}
    assume a stationary hardware failure hazard and a fixed
    write cost. In voluntary sharing, departure follows
    time-varying behavioral patterns (daily peaks at lunch
    and evening) and checkpoint cost itself depends on the
    bandwidth available at write time. Scheduled departures
    add a notice window $\tau$ in which the system must
    decide whether a final checkpoint can complete before
    reclaim, while emergency departures must fall back to
    the freshest completed snapshot. No prior model captures
    this joint time-varying hazard, network-limited cost,
    and reclaim-window feasibility.

  \item \textbf{Challenge~2 (Migration destination):}
    Campus networks are three-tier hierarchies with
    heterogeneous bandwidth: intra-building paths are
    $\approx$3.4$\times$ faster than the cross-building
    backbone in our measurements, and a destination that
    itself departs soon compounds failure risk.
    Topology-oblivious selection (round-robin, least-loaded)
    and conventional topology-aware placement (minimizing
    instantaneous transfer time) both miss the joint
    bandwidth-stability-deadline objective: a fast
    destination is useless if the checkpoint cannot land
    before the source's notice expires, and a nearby one
    is undesirable if its own reclaim hazard is high.

  \item \textbf{Challenge~3 (Migration traffic):}
    Voluntary reclaim imposes hard migration deadlines and
    instantaneous revocation events. Off-the-shelf rate
    control such as DSCP queuing or generic token-bucket
    shaping can cap throughput but cannot decide which
    migration flows are feasible under their notice windows
    or how to degrade gracefully when aggregate demand
    exceeds capacity. Commodity enforcement paths
    (\emph{e.g.}, iptables) further add $\sim$100\,ms
    latency, incompatible with sub-second reclaim events.
\end{itemize}

\begin{table}[t]
\centering
\caption{Why standard mechanisms are insufficient under
voluntary reclaim semantics.}
\label{tab:increment}
\scriptsize
\setlength{\tabcolsep}{3pt}
\begin{tabular}{p{0.19\columnwidth}p{0.24\columnwidth}p{0.24\columnwidth}p{0.23\columnwidth}}
\toprule
\textbf{Existing idea} & \textbf{Standard assumption} &
\textbf{Broken by VVN} & \textbf{\sys increment} \\
\midrule
Checkpointing & Stationary failures; fixed write cost &
Owner reclaim varies; cost depends on network &
Risk + bandwidth controller \\
Topology-aware placement & Destination is stable once selected &
Destination may be reclaimed soon; source has deadline &
Path + survival + deadline score \\
Traffic shaping & Fairness or rate cap is enough &
Flows have reclaim deadlines and infeasible overload cases &
Deadline admission with edge enforcement \\
Kill-switch & App action can enforce reclaim &
Provider needs immediate local autonomy &
Network-layer isolation semantics \\
\bottomrule
\end{tabular}
\end{table}

Voluntary departure thus turns workload migration into a
\emph{coupled control problem}: the checkpoint interval
determines how much state must be transferred; the
destination choice determines how much bandwidth is
available; and the traffic schedule determines whether that
transfer completes before the deadline. Each dimension
interacts with the others, and no existing protocol
addresses their joint requirements.

\textbf{Our contribution.}
We present \sys, which addresses these three challenges
through network-layer protocols grounded in analysis and
validated at scale.

\textbf{(P1) Reclaim-Aware Checkpoint Scheduling.}
P1 turns checkpointing from a fixed-period user policy into
an online control loop for voluntary reclaim. At each epoch
it uses the current emergency reclaim hazard $\lambda_e(t)$
and available bandwidth $B(t)$ (measured via eBPF) to select
a local interval:
\[
  \Delta t^{*}(t) \;=\; \sqrt{\frac{2C}{\lambda_e(t) \cdot B_{\text{eff}}(t)}}
\]
This local rule is only the per-job primitive. The protocol
contribution is the shared-network controller in
Theorem~\ref{thm:p1coupled}: when $K$ jobs contend for a
checkpoint bandwidth budget $B_{\mathrm{ckpt}}$, P1 minimizes
aggregate reclaim loss plus checkpoint overhead by allocating
bandwidth as
$b_i^*=B_{\mathrm{ckpt}}(\lambda_i C_i)^{1/3}/
\sum_j(\lambda_j C_j)^{1/3}$ and then setting each job's
interval from its assigned bandwidth. This gives P1 a
network-level coupling absent from independent per-job
timers. P1 also enforces a loss budget, checks whether
scheduled departures have enough notice for a final
checkpoint, and coordinates the resulting rate demands with
P3 (\S\ref{sec:p1}). A separate adaptivity-gap result
(Theorem~\ref{thm:p1corr}) explains when state-dependent
checkpointing improves over a state-independent fixed-interval
policy tuned to marginal averages: strict improvement occurs
iff $\lambda_e(t) B_{\text{eff}}(t)$ is not almost surely
constant, and the gap is a directly measurable
$L^{2}$ Cauchy--Schwarz defect of the joint
$(\lambda_e,B)$ trace.
P1 reduces average work loss by
\FILLNUM{66\%} compared to industry Slurm preempt+requeue
and \FILLNUM{38\%} versus Bamboo~\cite{thorpe2023bamboo},
and by \FILLNUM{6.7\%} relative to the
strongest state-independent online baseline (a time-varying
fixed-interval schedule), the latter matching the
$\sim$5.2\% Cauchy--Schwarz gap predicted by
Theorem~\ref{thm:p1corr} on the deployment trace
(\S\ref{sec:eval}).

\textbf{(P2) Volatility-Aware Destination Selection.}
P2 extends topology-aware scheduling with two constraints
absent from conventional placement: the source's notice
deadline and the destination's own reclaim hazard. It filters
candidate destinations that cannot receive the checkpoint
before the notice window closes, then scores feasible nodes
by current transfer cost and expected future re-migration
cost over the campus network graph. A locality filtering
lemma explains why same-building candidates dominate under
hierarchical campus bandwidths, while the full score prevents
choosing unstable local nodes (\S\ref{sec:p2}). P2 reduces
median migration downtime by \FILLNUM{38\%} versus random
destination selection.

\textbf{(P3) Deadline-Aware Migration Traffic Scheduling.}
P3 separates the protocol from the mechanism: the protocol
converts notice periods into per-flow bandwidth demands,
performs admission/allocation under a reserved research
traffic budget, and marks infeasible flows for degraded
recovery; Linux TC BPF token buckets are the edge
enforcement substrate. Under controlled-edge assumptions,
P3 provides two guarantees: feasible planned migrations
complete before their reclaim deadlines, and aggregate
migration traffic is bounded so non-migration traffic keeps
at least $B_{\min}$ bandwidth (\S\ref{sec:p3}). Kill-switch
enforcement moves to the network layer, reducing latency
from $\sim$100\,ms to under 1\,ms.

\textbf{Evaluation.}
We deploy \sys on a 54-node campus testbed with
heterogeneous GPUs (RTX~3090, RTX~4090, A100, A6000)
across four buildings (\S\ref{sec:eval}).
End-to-end results over two months show
\FILLNUM{28\%} higher GPU utilization, an
\FILLNUM{11.7 percentage-point} increase in migration success,
and under \FILLNUM{3\%}
research traffic degradation compared to fixed-interval
and topology-oblivious baselines,
with scheduling latency scaling sub-linearly up to
500 nodes in emulation.

The remainder of this paper is organized as follows.
\S\ref{sec:background} formalizes the voluntary volatile
network model. \S\ref{sec:related} reviews related work.
\S\ref{sec:measurement} presents our measurement study.
\S\ref{sec:p1}--\S\ref{sec:p3} describe the three
protocols. \S\ref{sec:impl} covers implementation.
\S\ref{sec:eval} presents evaluation.
\S\ref{sec:discussion} discusses limitations and concludes.

\section{Background and Motivation}
\label{sec:background}

\subsection{The Voluntary Volatile Network Model}
\label{sec:vvn}

We formalize the campus GPU sharing environment as a
\emph{Voluntary Volatile Network (VVN)}, characterized by
the following properties.

\begin{definition}[Voluntary Volatile Network]
A VVN is a distributed resource-sharing system in which:
(i) each node $v \in V$ is owned and operated by an
independent provider who may exit at any time, and
(ii) node departures are classified as \emph{scheduled}
(notice period $\tau > 0$, drawn from an empirical
distribution $F_\tau$ with $\tau_{\min}\geq 10$\,s in our
deployment) or \emph{emergency} ($\tau = 0$), with
emergency departures modeled as a \emph{non-homogeneous}
Poisson process with time-varying intensity
$\lambda_e(t)$.
\end{definition}

This model captures two key differences from data center
environments: \emph{departure is normative}, not a failure
to be avoided, and \emph{hardware is heterogeneous},
precluding low-level state migration across GPU
architectures.
It also induces a protocol contract not present in ordinary
clusters: a provider must be able to reclaim locally without
waiting for a central scheduler; emergency reclaim can only
recover from the freshest completed checkpoint; scheduled
reclaim succeeds only if the checkpoint transfer and restart
fit within the notice window; and migration traffic must not
consume the shared campus network simply because a provider
is exercising this right.

\subsection{Why Existing Approaches Fall Short}

Table~\ref{tab:comparison} summarizes how key systems
compare against the VVN requirements, and
Fig.~\ref{fig:architecture} previews \sys's coordinator
plus provider-node decomposition that we elaborate in
\S\ref{sec:p1}--\S\ref{sec:p3}.

\begin{table}[t]
\centering
\caption{Comparison of GPU/cluster management platforms
against VVN requirements.}
\label{tab:comparison}
\scriptsize
\resizebox{\columnwidth}{!}{
\begin{tabular}{lcccc}
\toprule
\textbf{System} & \textbf{Voluntary} & \textbf{Network-} & \textbf{Migration} & \textbf{Formal}\\
 & \textbf{Partic.} & \textbf{Adaptive} & \textbf{Protocol} & \textbf{Guarantees}\\
\midrule
Kubernetes~\cite{kubernetes}    & \ding{55} & \ding{55} & \ding{55} & \ding{55} \\
Slurm~\cite{yoo2003slurm}       & \ding{55} & \ding{55} & \ding{55} & \ding{55} \\
Gandiva~\cite{xiao2018gandiva}  & \ding{55} & \ding{55} & Partial   & \ding{55} \\
GPUnion~\cite{li2025gpunion}    & \ding{51} & \ding{55} & ALC only  & \ding{55} \\
\textbf{\sys}                   & \ding{51} & \ding{51} & P1+P2+P3  & \ding{51} \\
\bottomrule
\end{tabular}
}
\end{table}

\begin{figure*}[t]
  \centering
  \includegraphics[width=0.95\textwidth]{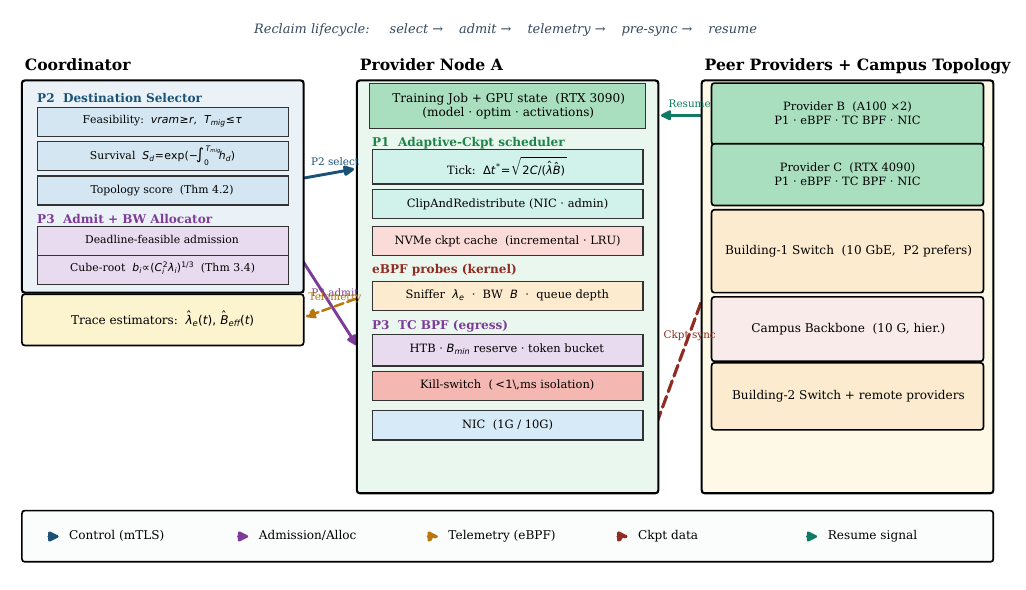}
  \caption{\sys architecture: coordinator (P2 selection + P3 admission), provider node (P1 Adaptive-Ckpt with eBPF/TC BPF enforcement), and topology-aware peer placement.}
  \label{fig:architecture}
\end{figure*}

\section{Related Work}
\label{sec:related}

\subsection{Campus and Volunteer Computing}

Condor~\cite{litzkow1988condor,thain2005condor} and volunteer
platforms such as SETI@home, Folding@home, and
BOINC~\cite{anderson2002seti,pande2003folding,anderson2004boinc}
show that idle, independently owned machines can be pooled, but
they target coarse-grained or stateless tasks rather than
long-running stateful GPU training. Volatile-grid systems such
as MPICH-V~\cite{bosilca2002mpichv} address different MPI-era
failure models. GPUnion~\cite{li2025gpunion} demonstrates that
provider-autonomy-first campus GPU sharing is feasible with
application-level checkpointing; \sys focuses on the missing
network-layer migration protocols required for reliable and
efficient operation under voluntary reclaim.

\subsection{Checkpoint and Migration Theory}

Young and Daly~\cite{young1974,daly2006} derive checkpoint
intervals under stationary hazards and fixed write costs.
Checkpoint libraries and user-level frameworks, including SCR,
FTI, libckpt, BLCR, and
DMTCP~\cite{moody2010scr,bautista2011fti,plank1995libckpt,hargrove2006blcr,ansel2009dmtcp},
optimize state capture and recovery under similar assumptions;
HPC and warehouse-scale failure studies likewise treat hazards
as quasi-stationary hardware properties~\cite{schroeder2007disk,pinheiro2007disk,schroeder2009dram}.
P1 retains Young/Daly as a local primitive but adds
time-varying reclaim hazard, network-dependent checkpoint cost,
notice-window feasibility, and multi-workload bandwidth
allocation.

VM live migration~\cite{clark2005live,hines2009postcopy,wood2011cloudnet}
and container checkpointing via CRIU reduce downtime in more
controlled settings, but they rely on hardware, kernel, or CUDA
compatibility that campus GPU pools cannot assume. \sys
therefore uses application-level checkpoints and focuses on
network-aware transfer timing, destination choice, and traffic
admission.

\subsection{GPU Cluster Scheduling}

GPU schedulers such as Gandiva, Tiresias, Themis, HiveD,
Optimus, Pollux, and AntMan~\cite{xiao2018gandiva,gu2019tiresias,mahajan2020themis,zhao2020hived,peng2018optimus,qiao2021pollux,xiao2020antman}
improve utilization, fairness, or elasticity in centrally managed
clusters; production traces reflect the same stable-tenancy
assumption~\cite{jeon2019analysis}. Pipeline and distributed
training systems~\cite{narayanan2019pipedream,huang2019gpipe,li2020pytorchddp}
and topology-aware co-design~\cite{wang2023topoopt} similarly
assume a non-volatile substrate. Spot-market schedulers address
cloud preemption~\cite{sharma2015spotcheck,harlap2017proteus},
and general-purpose managers such as Mesos, YARN, and
Borg~\cite{hindman2011mesos,vavilapalli2013yarn,verma2015borg}
manage offers and quotas. None of these systems model voluntary
provider reclaim or the network-coupled migration deadlines that
dominate campus VVNs.

\subsection{eBPF and Programmable Network Functions}

eBPF provides programmable end-host telemetry and enforcement
for packet processing and per-flow shaping~\cite{mccanne1993bpf,hoiland2018xdp,jouet2017bpfabric,saeed2017carousel}.
Datacenter transports and schedulers such as DCTCP, pFabric,
CONGA, PIAS, and Varys~\cite{alizadeh2010dctcp,alizadeh2013pfabric,alizadeh2014conga,bai2015pias,chowdhury2014varys}
optimize congestion or coflow completion in managed fabrics.
P3 does not introduce a new congestion-control primitive; it
adds reclaim-deadline admission/allocation and enforces those
decisions at provider edges during GPU migration.

\section{Measurement Study}
\label{sec:measurement}

\subsection{Deployment Setup}

We deployed a \sys prototype across 54 GPU nodes in four
buildings of a research-university campus over two months:
20 RTX~3090, 15 RTX~4090, 8 A100 (4 GPUs each), 10 A6000
(2 GPUs each), plus one coordinator. Nodes span access-
and distribution-layer segments on a three-tier campus
network (1\,Gbps access, 1--10\,Gbps distribution, 10\,Gbps
core). Provider agents logged departure events, checkpoint
sizes, and per-flow bandwidth samples from eBPF TC BPF
programs attached to each node's egress interface.

\subsection{Findings}

\textbf{F1: Departure patterns are bursty and predictable.}
Emergency (no-notice) departures account for 42\% of all
events, concentrated in two daily windows
(11:30--13:00 and 17:30--20:00); $\lambda_e(t)$ varies up
to $6\times$ between peak and off-peak hours
(Fig.~\ref{fig:departure_rate}). This time-of-day structure
motivates P1's non-homogeneous Poisson model: a
fixed-interval strategy tuned for peak hours
over-checkpoints during quiet periods by up to
$\sim$4.1$\times$.
\textbf{F2: Checkpoint sizes span three orders of magnitude.}
Payloads range from $\sim$120\,MB (ResNet-50 dense
state-dict) to 48\,GB (full-precision LLM state-dicts with
optimizer state), with a 3.2\,GB median across the mixed
workload of small models, PEFT/LoRA fine-tuning, and
full-state training (Fig.~\ref{fig:ckpt_size_cdf});
1\,Gbps-link transmission ranges from $<$2\,s to $>$6\,min,
showing that a topology-unaware selector routing across the
10\,Gbps core can be \emph{slower} than a same-building hop
when the core is loaded.

\begin{figure}[t]
  \centering
  \begin{subfigure}[t]{0.48\columnwidth}
    \includegraphics[width=\linewidth]{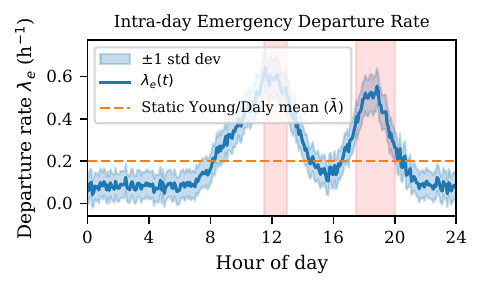}
    \caption{$\lambda_e(t)$ over two months.}
    \label{fig:departure_rate}
  \end{subfigure}\hfill
  \begin{subfigure}[t]{0.48\columnwidth}
    \includegraphics[width=\linewidth]{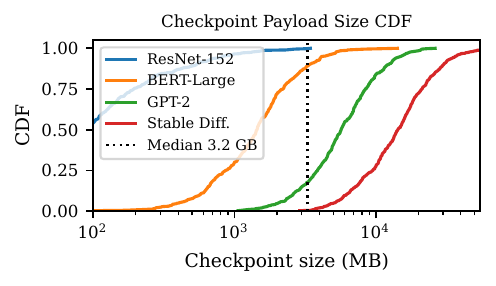}
    \caption{Checkpoint payload CDF.}
    \label{fig:ckpt_size_cdf}
  \end{subfigure}
  \caption{Measurement findings: (a) intra-day reclaim hazard and (b) checkpoint-size spread motivate adaptive intervals and topology-aware destination selection.}
  \label{fig:measure}
\end{figure}

\textbf{F3: Bandwidth is highly dynamic.}
eBPF measurements show intra-building $B_{\text{local}}$
averaging 820\,Mbps versus core-crossing $B_{\text{core}}$
240\,Mbps at peak ($r\approx 3.4$); per-link $B(t)$
fluctuates by $\pm$60\% over 10-minute windows, requiring
real-time measurement for accurate cost prediction in both
P1 and P2.
\textbf{F4: Hazard and bandwidth are anti-correlated.}
Peak reclaim hours (lunch, evening) coincide with the worst
campus bandwidth because human activity drives both:
$\varrho_{\lambda,B}=\mathrm{corr}(\lambda_e,B_{\text{eff}})
=-0.43$ (95\% bootstrap CI $[-0.51,-0.34]$, $n{=}6048$
buckets), with $\mathrm{corr}(\lambda_e,1/B_{\text{eff}})
=+0.39$ and CVs $\mathrm{CV}_\lambda=0.62$,
$\mathrm{CV}_{1/B}=0.31$. This violates the independence
assumption underlying classical periodic-checkpoint
analyses~\cite{young1974,daly2006}, which optimize against
$\mathbb{E}[\lambda]$ and $\mathbb{E}[B]$ separately and
under-estimate the cost of long intervals when $\lambda$ and
$1/B$ co-vary positively. P1 exploits this joint structure
as its core algorithmic increment over fixed-interval
schedules (\S\ref{sec:p1corr}, Thm.~\ref{thm:p1corr}).
Together, F1--F4 plus the observation of up to 7 simultaneous
departures drive the three protocol designs: F1's $\lambda_e(t)$
dynamics motivate P1's time-varying interval; F2's payload
spread motivates P2's topology-aware selector; concurrent
departures motivate P3's concurrency control.

\section{P1: Reclaim-Aware Checkpoint Scheduling}
\label{sec:p1}

\subsection{Reclaim-Aware System Model}

P1 separates the two reclaim modes. Emergency reclaim
($\tau=0$) cannot rely on the source remaining available,
so recovery uses the freshest checkpoint that has already
completed. Scheduled reclaim ($\tau>0$) first tests final
checkpoint feasibility: a zero-loss handoff is possible only
if $C/B_{\text{eff}}(t)+T_r \leq \tau$, where $C$ is the
checkpoint payload size, $T_r$ is restart time, and
$B_{\text{eff}}(t)=\min(W,B(t))$ is the effective transfer
throughput from local write speed $W$ and measured network
bottleneck $B(t)$.

For emergency protection, P1 treats $\lambda_e(t)$ and
$B_{\text{eff}}(t)$ as piecewise-stationary over one control
epoch. This is the operational assumption behind the closed
form: the local interval is recomputed at
each checkpoint, not a claim that a single interval is globally
optimal for an arbitrary non-homogeneous process. The
per-unit-time expected slowdown cost of a checkpoint interval
$\Delta t$ is:
\begin{equation}
  \mathcal{L}(\Delta t, t)
  = \underbrace{\lambda_e(t) \cdot \frac{\Delta t}{2}}_{\text{expected work loss}}
  + \underbrace{\frac{C}{B_{\text{eff}}(t) \cdot \Delta t}}_{\text{checkpoint I/O overhead}}
  \label{eq:p1cost}
\end{equation}
The first term grows linearly in $\Delta t$ (more work
lost per departure when intervals are long) and the second
decays as $1/\Delta t$ (checkpointing overhead amortized
over longer intervals).

\subsection{Local Checkpoint Interval}

\begin{lemma}[Local Reclaim-Aware Optimal Interval]
\label{lem:p1local}
Within a control epoch where $\lambda_e(t)$ and
$B_{\text{eff}}(t)$ are treated as constant estimates, the
checkpoint interval that minimizes
$\mathcal{L}(\Delta t, t)$ in \eqref{eq:p1cost} is:
\begin{equation}
  \Delta t^*(t) = \sqrt{\frac{2C}{\lambda_e(t) \cdot B_{\text{eff}}(t)}}
  \label{eq:p1opt}
\end{equation}
with optimal cost $\mathcal{L}^* = \sqrt{2\lambda_e C /
B_{\text{eff}}}$.
\end{lemma}

\begin{proof}
See Appendix~\ref{app:proofs}.
\end{proof}

\noindent
Lemma~\ref{lem:p1local} is the local control law inside P1.
It generalizes the classical periodic-checkpoint
formula~\cite{young1974} to the dynamic network-constrained
setting: $B_{\text{eff}}(t)$ replaces the implicit constant
write-speed, and $\lambda_e(t)$ is re-estimated from observed
reclaim behavior rather than assumed as a stationary hardware
failure rate.

\subsection{Network-Coupled Checkpoint Bandwidth Allocation}

The local rule is insufficient when multiple jobs share a
campus bottleneck: if $K$ checkpoints start together, the
bandwidth each job sees is an allocation decision rather than
an exogenous constant. P1 therefore solves a per-epoch
bandwidth-allocation problem before instantiating the
intervals in Lemma~\ref{lem:p1local}. Let $b_i$ be the
checkpoint bandwidth assigned to job $i$, with
$\sum_i b_i \leq B_{\mathrm{ckpt}}$, where
$B_{\mathrm{ckpt}}$ is the checkpoint bandwidth budget after
P3's research-traffic reservation. For job $i$ with
emergency reclaim hazard $\lambda_i$ and checkpoint payload
$C_i$, the local cost under assigned bandwidth $b_i$ is:
\begin{equation}
  \mathcal{L}_i(\Delta_i,b_i)
  = \frac{\lambda_i\Delta_i}{2}
  + \frac{C_i}{b_i\Delta_i}.
  \label{eq:p1multi-cost}
\end{equation}

\begin{theorem}[Network-Coupled Checkpoint Allocation]
\label{thm:p1coupled}
Within a control epoch, assume $K$ active checkpointing jobs
share a checkpoint bandwidth budget $B_{\mathrm{ckpt}}$ and
that no local write-speed cap is binding. The allocation and
intervals that minimize
\begin{equation*}
  \sum_{i=1}^{K}\mathcal{L}_i(\Delta_i,b_i)
  \quad
  \text{s.t.}\quad
  \sum_{i=1}^{K} b_i \leq B_{\mathrm{ckpt}},\; b_i>0
\end{equation*}
are:
\begin{equation}
  b_i^*
  =
  B_{\mathrm{ckpt}}\,
  \frac{(\lambda_i C_i)^{1/3}}
       {\sum_{j=1}^{K}(\lambda_j C_j)^{1/3}},
  \label{eq:p1cube}
\end{equation}
and
\begin{equation}
  \Delta_i^*
  =
  \sqrt{\frac{2C_i}{\lambda_i b_i^*}}.
  \label{eq:p1multi-delta}
\end{equation}
The minimized aggregate cost is
$\sqrt{2/B_{\mathrm{ckpt}}}\,
(\sum_i(\lambda_iC_i)^{1/3})^{3/2}$.
\end{theorem}

\begin{proof}
See Appendix~\ref{app:proofs}.
\end{proof}

\noindent
Theorem~\ref{thm:p1coupled} is the P1 protocol rule that is
absent from independent per-job timers: checkpoint
bandwidth is not split equally, nor linearly by checkpoint
size. The cube-root allocation gives more bandwidth to jobs
with higher reclaim hazard and larger checkpoints, but with
diminishing returns, so a single large or risky job cannot
starve the rest.

\begin{lemma}[Box-Constrained Allocation]
\label{lem:p1box}
Let $W_i$ denote a per-job write-speed or administrator cap
on $b_i$, and consider
$\min\sum_i\sqrt{2\lambda_iC_i/b_i}$ s.t.\
$\sum_i b_i\leq B_{\mathrm{ckpt}}$, $0<b_i\leq W_i$.
The following water-filling iteration terminates in at most
$K$ steps and returns the global optimum: (1) solve the
unconstrained problem via~\eqref{eq:p1cube}; (2) for every
job whose unconstrained $b_i^*>W_i$, fix $b_i=W_i$, subtract
$W_i$ from $B_{\mathrm{ckpt}}$, and remove $i$ from the
active set; (3) re-solve over the residual active set and
residual budget; repeat until no cap is violated.
\end{lemma}

\begin{proof}[Proof sketch]
See Appendix~\ref{app:proofs}.
\end{proof}

If a local write cap $W_i$ or an
administrator cap is binding, P1 invokes Lemma~\ref{lem:p1box};
the same constrained clipping used in
Eq.~\eqref{eq:p1constrained} below also enforces loss
budgets and scheduled-reclaim deadlines.

\subsection{Adaptivity Gap of State-Dependent Checkpointing}
\label{sec:p1corr}

Let $(\Lambda, \Theta)$ denote the joint stationary
distribution of $(\lambda_e(t), 1/B_{\text{eff}}(t))$ over an
operational horizon, with $\lambda_e, B_{\text{eff}}>0$ a.s.\
and finite first moments. Define two policy classes:
\emph{state-dependent} $\pi_{\text{dep}}$ samples the
current $(\lambda_e, B)$ at each epoch and applies
Lemma~\ref{lem:p1local}'s closed form (in multi-workload
epochs, P1 first applies Theorem~\ref{thm:p1coupled} and
treats the assigned $b_i$ as the job's effective bandwidth);
\emph{state-independent} $\pi_{\text{ind}}$ commits to a
single $\Delta t_0$ chosen optimally from the marginals
$(\bar\lambda_e, \overline{1/B})$, the \emph{best}
a stationary periodic-checkpoint analysis can achieve, even
with full distributional knowledge, because its cost model
factors through $\mathbb{E}[\lambda]$ and $\mathbb{E}[1/B]$
only.

\begin{theorem}[Adaptivity Gap]
\label{thm:p1corr}
Let $W := \sqrt{\lambda_e/B_{\text{eff}}}$. The expected
per-unit-time cost ratio satisfies
\begin{equation}
  \frac{\mathbb{E}[\mathcal{L}_{\text{ind}}^*]}{\mathbb{E}[\mathcal{L}_{\text{dep}}^*]}
  \;=\;
  \frac{\sqrt{\,\mathbb{E}[\lambda_e]\,\mathbb{E}[1/B_{\text{eff}}]\,}}
       {\mathbb{E}\!\left[\sqrt{\lambda_e/B_{\text{eff}}}\,\right]}
  \;\geq\; 1,
  \label{eq:gap-exact}
\end{equation}
where the numerator is the square root of the product
$\mathbb{E}[\lambda_e]\,\mathbb{E}[1/B_{\text{eff}}]$ and the
denominator is the expectation of the random variable
$W:=\sqrt{\lambda_e/B_{\text{eff}}}$; the inequality is exactly
the Cauchy--Schwarz bound on $(\sqrt{\lambda_e},\sqrt{1/B_{\text{eff}}})$
in $L^{2}(\mathbb{P})$.
with equality iff $\lambda_e B_{\text{eff}}$ is a.s.\
constant. Equivalently, the absolute gap is
\begin{equation}
  \mathbb{E}[\mathcal{L}_{\text{ind}}^*]
  - \mathbb{E}[\mathcal{L}_{\text{dep}}^*]
  \;=\; \sqrt{2C}\cdot
  \frac{\Delta_{\mathrm{CS}}}
       {\sqrt{\bar\lambda\,\bar\Theta} + \mathbb{E}[W]},
  \label{eq:gap-quant}
\end{equation}
where $\bar\Theta:=\mathbb{E}[1/B_{\text{eff}}]$ and the
$L^{2}(\mathbb{P})$ Cauchy--Schwarz defect of
$(\sqrt{\lambda_e},\sqrt{1/B_{\text{eff}}})$ is
\begin{equation}
  \Delta_{\mathrm{CS}}
  \;:=\; \mathbb{E}[\lambda_e]\cdot\mathbb{E}[1/B_{\text{eff}}]
  \;-\; \big(\mathbb{E}[W]\big)^{2}
  \;\geq\; 0.
  \label{eq:cs}
\end{equation}
\end{theorem}

\begin{proof}
See Appendix~\ref{app:proofs}.
\end{proof}

\noindent
\textbf{Interpretation.} A stationary fixed-interval analysis
treats $\lambda$ and $B$ as scalar nominal values, so
$\Delta_{\mathrm{CS}}\equiv 0$ by construction and the
adaptivity gap is structurally invisible.
Theorem~\ref{thm:p1corr} expresses the benefit of
state-dependent interval selection as a directly
\emph{measurable} functional
$(\bar\lambda,\bar\Theta,\mathbb{E}[W])$ of the deployment
trace, observable from any deployment that estimates
$\lambda_e$ and $B_{\text{eff}}$ online. Plugging the
two-month calibration statistics into~\eqref{eq:gap-exact}
predicts a $\sim$5.2\% gap, which we verify empirically
against baseline B7 in Table~\ref{tab:main_results}. A
detailed binding-regime analysis (\textit{B}-binding vs.\
\textit{W}-binding epochs) appears in
Appendix~\ref{app:protocol-details}.

\subsection{Constrained Solution and Online Algorithm}

Three reclaim-induced constraints bound the unconstrained
optimum after Theorem~\ref{thm:p1coupled} assigns
checkpoint bandwidth $b_i$:
\emph{network feasibility} ($\Delta_i \geq C_i/(\beta b_i)$,
$\beta{=}0.2$; see below), \emph{loss budget}
($\Delta_i \leq 2L_{\max,i}$), and \emph{scheduled-reclaim
deadline} ($C_i/b_i+T_{r,i}\leq\tau_i$ for notice
$\tau_i$). The constrained per-job solution is
\begin{equation}
  \Delta_{i,c}^*(t) =
  \text{clip}\!\left(\sqrt{\frac{2C_i}{\lambda_i(t)b_i}},\;
  \frac{C_i}{\beta b_i},\; 2L_{\max,i}\right).
  \label{eq:p1constrained}
\end{equation}

\noindent\textbf{Operational meaning of $\beta$.}
The floor $C_i/(\beta b_i)$ enforces a budget split on the
allocated egress slice $b_i$: at most a fraction $\beta$ of
$b_i$ is consumed by checkpoint writes, leaving $(1-\beta)b_i$
for live gradient/activation traffic of co-located training
jobs. We set $\beta=0.2$ so that checkpoint I/O never steals
more than $20\%$ of the slice, keeping all-reduce tail
latency within the SLO window observed in
\S\ref{sec:measurement} (F2--F3). Smaller $\beta$ throttles
checkpoint frequency (raising recompute loss after
preemption); larger $\beta$ inflates communication tail and
violates fairness. The choice is robust: a sensitivity sweep
over $\beta\in\{0.1,0.2,0.3\}$ in
Appendix~\ref{app:eval-extra} changes end-to-end
goodput by $<2\%$, with $\beta=0.2$ giving the best
loss--throughput trade-off.
The online controller \textsc{Adaptive-Ckpt}
(Algorithm~\ref{alg:p1})
estimates $\lambda_e(t)$ via sliding-window MLE over the
departure log and samples $B(t)$ from eBPF egress counters
at 100\,ms intervals; deadline checks use a lower-confidence
bound $\hat B^-$.

\begin{algorithm}[t]
\caption{\textsc{Adaptive-Ckpt}: network-coupled P1 controller}
\label{alg:p1}
\small
\begin{algorithmic}[1]
\Require active jobs $\mathcal{J}$ with $(C_i,L_{\max,i},\hat\lambda_i)$;
  link budget $\hat B^{-}_{\text{ckpt}}(t)$; reclaim notice $\tau_i$
\Procedure{Tick}{$t$}
  \State $S \gets \sum_{j\in\mathcal{J}}(\hat\lambda_j^{+}(t)\,C_j)^{1/3}$
        \Comment{cube-root coupling, Thm.~\ref{thm:p1coupled}}
  \ForAll{$i\in\mathcal{J}$}
    \State $b_i \gets \hat B^{-}_{\text{ckpt}}(t)\cdot
            (\hat\lambda_i^{+}(t)\,C_i)^{1/3}/S$
    \State $b_i \gets \textsc{ClipAndRedistribute}(b_i)$
           \Comment{NIC \& admin caps}
    \If{$\tau_i \neq \bot \wedge C_i/b_i + T_{r,i} \leq \tau_i$}
      \State \textbf{emit} final-ckpt$(i)$; \textbf{continue}
    \EndIf
    \State $\Delta_i \gets
      \mathrm{clip}\!\left(\sqrt{2C_i/(\hat\lambda_i^{+}(t)\,b_i)},\,
      \tfrac{C_i}{\beta b_i},\,2L_{\max,i}\right)$
    \State \textbf{export} $(i,C_i,b_i,\Delta_i,L_{\max,i})\to$ P3
  \EndFor
\EndProcedure
\Loop\, every control period: \Call{Tick}{$t$}
\EndLoop
\end{algorithmic}
\end{algorithm}

\subsection{Multi-Workload Coordination}

When $K$ jobs are co-resident, independent checkpoint timers
can synchronize and collapse per-job bandwidth to roughly
$B(t)/K'$ ($K'\leq K$ active transferrers). Theorem~\ref{thm:p1coupled}
prevents this by converting the shared bottleneck into
per-job rate demands before timers fire. P1 exports each
checkpoint as a demand tuple
$(i,C_i,b_i,\Delta_i,L_{\max,i})$ to P3, which staggers
non-urgent checkpoints, reserves deadline bandwidth for
scheduled reclaim events, and rejects infeasible overload
sets into degraded recovery (\S\ref{sec:p3}). This
allocation-plus-admission path is the protocol difference
from independent fixed-interval timers.

\section{P2: Volatility-Aware Destination Selection}
\label{sec:p2}

\subsection{Network Model}

We model the campus network as a weighted graph
$G=(V,E,w)$, where $V$ includes both servers and switches,
and $w(e,t)$ is the available bandwidth on link $e$ at
time $t$, measured passively by eBPF TC BPF programs
installed on each provider node.

For a scheduled reclaim at source $s$ with notice window
$\tau_s$, destination selection has a hard feasibility
constraint in addition to a cost objective. The migration
transfer time from source $s$ to destination $d$ is:
\begin{equation}
  T_{\text{mig}}(s,d,t) =
  \frac{C}{\text{bw}_{\text{bot}}(s,d,t)} + T_r
\end{equation}
where $\text{bw}_{\text{bot}}(s,d,t)=\min_{e\in
P(s,d)}w(e,t)$ is the path bottleneck and $T_r$ is the
container restart time. A destination is deadline-feasible
only if $T_{\text{mig}}(s,d,t) \leq \tau_s$; otherwise P2
excludes it for zero-loss handoff and marks the job for
degraded recovery from the latest completed checkpoint.

\subsection{Joint Scoring Function}

Minimizing only $T_{\text{mig}}$ is myopic: in a VVN,
the destination is itself an autonomous provider. If $d$
has high future reclaim hazard, the workload will
re-migrate shortly, incurring additional overhead. P2
therefore combines path cost, deadline feasibility, and
destination survival. Let
$S_d(T_{\text{rem}})=\exp(-\int_t^{t+T_{\text{rem}}}
h_d(u)\,du)$ be the probability that destination $d$
survives for the job's remaining runtime under estimated
hazard $h_d$. We minimize the one-step lookahead cost:
\begin{equation}
  \text{Score}(d) = T_{\text{mig}}(s,d,t)
  + \alpha \cdot (1-S_d(T_{\text{rem}}))
    \cdot \bar{T}_{\text{mig}}(d)
  \label{eq:p2score}
\end{equation}
where $T_{\text{rem}}$ is the estimated remaining job
runtime, $\bar{T}_{\text{mig}}(d)$ is the average
migration cost from $d$ to other candidates, and
$\alpha \in [0.5,2.0]$ is a (dimensionless) stability
penalty weight. We calibrate $\alpha$ once per deployment
by minimizing the held-out work loss of the P1+P2+P3 stack
on weeks 1--2 of the trace over the grid
$\{0.5,1.0,1.5,2.0\}$; the empirical minimum is $\alpha=1.0$
on our deployment. Re-calibration is only needed when the
ratio $B_{\text{local}}/B_{\text{core}}$ or the
inter-building hazard heterogeneity changes by more than
$\sim$30\%, an event we did not observe in the held-out
period (\S\ref{sec:eval}-E8).

The optimal destination is:
$d^* = \arg\min_{d \in \mathcal{D}_c(s)} \text{Score}(d)$,
where $\mathcal{D}_c(s)$ is the set of compatible
candidates satisfying VRAM, CUDA capability, load, and
deadline feasibility constraints.

\subsection{Locality Filtering Lemma}

\begin{lemma}[Hierarchical Locality Filter]
\label{thm:p2}
Assume a hierarchical campus window in which every
same-building feasible path from $s$ has bottleneck at least
$B_{\text{local}}^{-}$ and every cross-building path has
bottleneck at most $B_{\text{core}}^{+}$, with
$B_{\text{local}}^{-}\geq B_{\text{core}}^{+}$. If a
same-building candidate satisfies GPU and deadline
constraints, then the best same-building candidate has
transfer time no larger than the best cross-building
candidate for the current migration.
\end{lemma}

\begin{proof}[Proof sketch]
See Appendix~\ref{app:proofs}.
\end{proof}

\noindent
With a measured median ratio
$B_{\text{local}}/B_{\text{core}}\approx 3.4$ in
\S\ref{sec:measurement}, the lemma justifies using locality
as a candidate filter rather than a final decision rule. The
technical increment over conventional topology-aware
scheduling is the subsequent survival and deadline scoring:
P2 will reject a local node that is unstable or cannot receive
the checkpoint before $\tau_s$.

\subsection{Algorithm: TopoAware-Select}

P2's destination selector (Algorithm~\ref{alg:p2})
first applies the
feasibility filter (VRAM, capability, load, deadline),
then promotes same-building candidates via
Lemma~\ref{thm:p2} (falling back to the cross-building set
when fewer than $k_{\min}$ local options exist), and finally
scores each remaining destination by
$\sigma_d = C/\text{bw}_d + T_r +
\alpha(1-S_d)\,\bar{T}_{\text{mig}}(d)$ where $S_d$ is the
survival probability over the remaining runtime. The total
cost is $O(|\mathcal{D}_{\text{cand}}|\cdot L)$ with
$L\leq 3$ (three-tier hop count)---negligible compared to
the transfer itself. Bandwidth estimates are refreshed from
the shared eBPF measurement table on each invocation.

\begin{algorithm}[t]
\caption{\textsc{TopoAware-Select}: P2 destination scoring}
\label{alg:p2}
\small
\begin{algorithmic}[1]
\Require src $s$; spec $(C,r,c_{\min})$; remaining runtime $T_{\text{rem}}$;
  notice $\tau_s$; penalty $\alpha$
\Ensure destination $d^{*}$ or \textsc{fail}
\Function{Select}{$s, C, r, c_{\min}, T_{\text{rem}}, \tau_s$}
  \State $\mathcal{D}_c \gets \{d \neq s:
    \text{vram}(d)\!\geq\! r,\,
    \text{cap}(d)\!\geq\! c_{\min},\,
    \text{load}(d)\!\leq\!\theta,\,
    C/\text{bw}(s,d)+T_r\leq\tau_s\}$
  \If{$\mathcal{D}_c = \emptyset$} \Return \textsc{fail}
  \EndIf
  \State $\mathcal{D}_{\ell} \gets \{d\in\mathcal{D}_c:
    \text{same-bldg}(s,d)\}$
  \State $\mathcal{D} \gets \mathcal{D}_{\ell}$
    \textbf{if} $|\mathcal{D}_{\ell}| \geq k_{\min}$
    \textbf{else} $\mathcal{D}_c$
  \ForAll{$d \in \mathcal{D}$}
    \State $\text{bw}_d \gets \textsc{eBPF-Bottleneck}(s,d)$
    \State $S_d \gets \textsc{Survival}(d, T_{\text{rem}})$
    \State $\sigma_d \gets C/\text{bw}_d + T_r
       + \alpha(1-S_d)\,\bar{T}_{\text{mig}}(d)$
  \EndFor
  \State \Return $\arg\min_{d\in\mathcal{D}}\sigma_d$
\EndFunction
\end{algorithmic}
\end{algorithm}

\begin{figure}[t]
  \centering
  \includegraphics[width=\columnwidth]{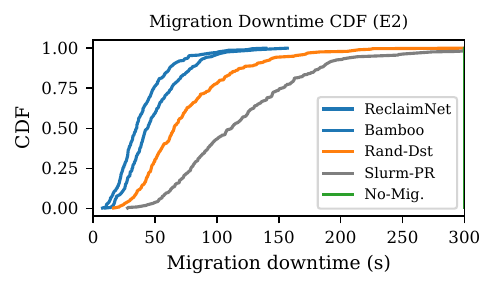}
  \caption{Per-migration downtime CDF: \sys cuts median by 38\% and p99 by 52\% versus Random-Dst.}
  \label{fig:p2_cdf}
\end{figure}

\section{P3: Deadline-Aware Migration Traffic Scheduling}
\label{sec:p3}

\subsection{Traffic Classification}

P3 classifies all traffic into four priority classes
using DSCP markings applied by TC BPF programs at the
egress of each provider node:

\begin{itemize}
  \item \textbf{$P_{\text{hi}}$ (Emergency migration):}
    Emergency reclaim with $\tau=0$ cannot assume the source
    remains available long enough to send new state. The job
    resumes from the freshest completed checkpoint; only
    small metadata/notification traffic is marked DSCP EF (46).
    If a provider grants a short grace period, P3 treats it as
    an urgent planned migration with a small $\tau$.
  \item \textbf{$P_{\text{med}}$ (Planned migration):}
    Provider departure with notice period $\tau>0$;
    subject to deadline-aware admission and allocation, marked
    DSCP AF41 (34).
  \item \textbf{$P_{\text{lo}}$ (Pre-sync):}
    Incremental background pre-synchronization from
    Algorithm~\ref{alg:p1}; marked DSCP CS1 (8).
  \item \textbf{$P_{\text{bg}}$ (Research traffic):}
    All other flows (SSH, Jupyter, data downloads);
    marked Best Effort (0).
\end{itemize}

\subsection{Admission and Bandwidth Allocation}

Total available bandwidth $B$ is split into a migration
share $B_{\text{mig}} = B - B_{\text{res}}$ and a
reserved share for research traffic:
$B_{\text{res}} = \max(\beta B,\ B_{\min})$,
with $\beta = 0.3$ calibrated from measurements and
$B_{\min}$ set by the campus administrator.

For each planned migration $k$, P3 converts the provider's
notice period into a minimum bandwidth demand:
\begin{equation}
  r_k^{\min}=\frac{C_k}{\tau_k-T_{r,k}} .
\end{equation}
The set is fully feasible when
$\sum_k r_k^{\min}\leq B_{\text{mig}}$. In that case, P3
admits all flows, assigns each $r_k^{\min}$, and distributes
remaining bandwidth with max-min water filling. When the
set is infeasible, P3 admits flows by earliest deadline first
subject to the same research-traffic reservation, and marks
the rest for degraded recovery from their freshest completed
checkpoints. This admission step is the protocol component;
the token bucket only enforces the selected rates.

Rates are written to a per-flow BPF hash map consumed by
the TC BPF token-bucket program.

Staggered scheduling offsets the start of planned flows
by $\delta = T_{\text{notice}}/K'$ each to eliminate
synchronized TCP slow-start bursts.

\subsection{eBPF Implementation}

A TC BPF program attached to each node's egress
interface implements per-flow token-bucket rate limiting
using \texttt{BPF\_MAP\_TYPE\_HASH}. Each bucket entry
stores the allocated rate (bytes/s), token count, and
last-refill timestamp. On each packet, tokens are
replenished proportionally to elapsed time; packets
are forwarded if the token count exceeds packet length,
otherwise dropped (TCP's congestion control handles
retransmission, achieving smooth rate enforcement
without kernel scheduler changes).

\subsection{Theorems 3 and 4: Formal Guarantees}

\begin{theorem}[Migration Completion Guarantee]
\label{thm:p3completion}
For an admitted planned migration $k$ with notice
$\tau_k>T_{r,k}$, if P3 assigns rate
$r_k \geq C_k / (\tau_k - T_{r,k})$ and the controlled
edge path sustains that rate, then the checkpoint transfer
and restart complete before provider departure.
\end{theorem}

\begin{proof}
See Appendix~\ref{app:proofs}.
\end{proof}

\begin{theorem}[Research Traffic Isolation on Controlled Access Bottleneck]
\label{thm:p3isolation}
Consider a single access-layer bottleneck of capacity $B$
whose upstream traffic originates entirely from provider
nodes running the \sys agent. Assume all migration flows
crossing this bottleneck are classified by P3 and enforced
at participating provider egress points. Then P3 bounds
aggregate migration traffic on the bottleneck to at most
$B-B_{\min}$, leaving at least $B_{\min}$ capacity for
non-migration traffic on that bottleneck.
\end{theorem}

\begin{proof}
See Appendix~\ref{app:proofs}.
\end{proof}

\noindent
\textbf{Scope of the guarantee.} Theorem~\ref{thm:p3isolation}
applies per access bottleneck and only over flows whose source
egresses are controlled by a \sys agent. It does \emph{not}
extend to distribution- or core-layer links shared with
non-participating tenants, and partial deployments (a subset
of providers without P3) weaken the guarantee proportionally
to the uncontrolled migration share. In our 54-node deployment
all providers run the agent; we discuss the partial-deployment
regime in \S\ref{sec:discussion}.

\subsection{Kill-Switch Network Enforcement}

When a provider activates the kill-switch, the provider
agent sets a flag in a BPF array map. A TC BPF ingress
program checks this flag on every inbound packet:
new SYN segments are dropped immediately, isolating
the node from new work; in-flight migration transfers
(already past SYN) are allowed to complete only when the
provider selected scheduled reclaim or a grace period,
preserving checkpoint integrity without blocking the
provider's local reclaim decision. End-to-end kill-switch
latency is
under 1\,ms---two to three orders of magnitude lower
than application-layer enforcement paths ($\sim$100\,ms)
and iptables ($\sim$10\,ms).

\begin{figure}[t]
  \centering
  \includegraphics[width=\columnwidth]{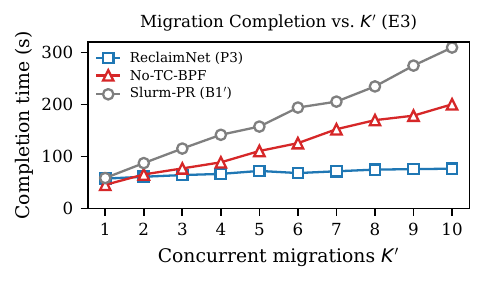}
  \caption{Migration completion time vs.\ concurrent migrations $K'$: P3 keeps admitted flows within their notice windows while baselines degrade linearly.}
  \label{fig:p3_completion}
\end{figure}

\begin{figure}[t]
  \centering
  \includegraphics[width=\columnwidth]{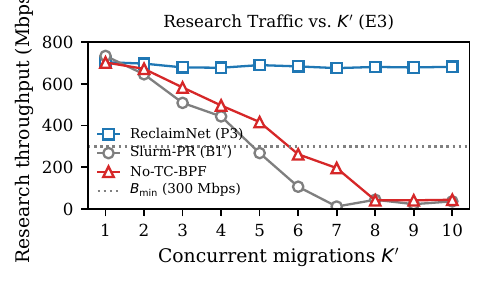}
  \caption{Research traffic throughput vs.\ $K'$: P3 sustains the $B_{\min}$ reserve while No-TC-BPF collapses to near zero.}
  \label{fig:p3_traffic}
\end{figure}

\section{Implementation}
\label{sec:impl}

\sys is implemented in $\sim$5{,}400 LoC across three layers:
\textbf{node agents} ($\sim$1{,}500 LoC C++17) drive
container lifecycle, checkpoint triggers, P2 selection, and
BPF-map updates; \textbf{eBPF programs} ($\sim$800 LoC C,
Clang 17 + libbpf) realise the P3 TC BPF token bucket and
the kill-switch; and the \textbf{coordinator}
($\sim$2{,}000 LoC C++17) hosts P3 admission/allocation,
the P2 topology database, and P1's departure-rate
estimator. mTLS-1.3 with short-lived coordinator-issued
JWTs authenticates all control traffic; A100/A6000 nodes
use MIG hardware partitioning while RTX 3090/4090 nodes
zero VRAM pages on container teardown. A full software-stack
and security-mechanism description appears in
Appendix~\ref{app:impl}.

\section{Evaluation}
\label{sec:eval}

\begin{figure*}[!t]
  \centering
  \begin{subfigure}[t]{\textwidth}
    \centering
    \includegraphics[width=\textwidth]{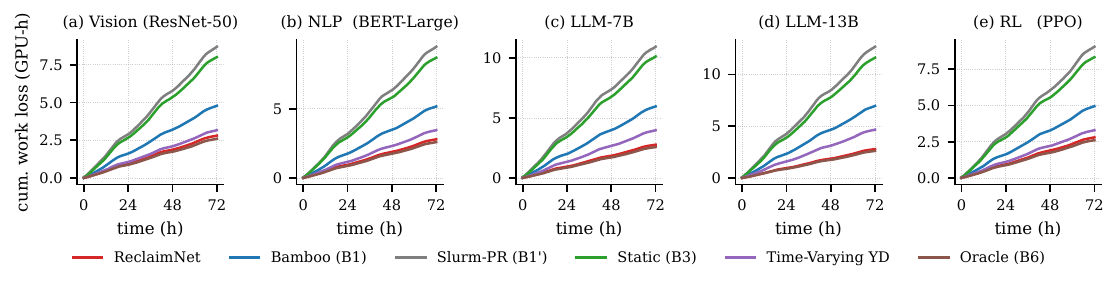}
    \caption{Cumulative work loss across five representative
      job classes (72-h slice of the two-month deployment):
      (a)~Vision (ResNet-50, $\sim$120\,MB dense state),
      (b)~NLP (BERT-Large, $\sim$340\,MB),
      (c)~LLM-7B with LoRA fine-tuning
      ($\sim$720\,MB adapter+optimizer),
      (d)~LLM-13B with LoRA ($\sim$1.2\,GB),
      (e)~RL (PPO, $\sim$210\,MB).
      Full-state LLM training (multi-GB to 48\,GB payloads, cf.\
      Fig.~\ref{fig:ckpt_size_cdf}) is shown in the
      Appendix~\ref{app:eval-extra} stress-test panel.
      \sys tracks the Oracle (B6) lower bound within 8--12\%
      on every job class.}
    \label{fig:workloss_jobs}
  \end{subfigure}

  \vspace{4pt}
  \begin{subfigure}[t]{\textwidth}
    \centering
    \includegraphics[width=\textwidth]{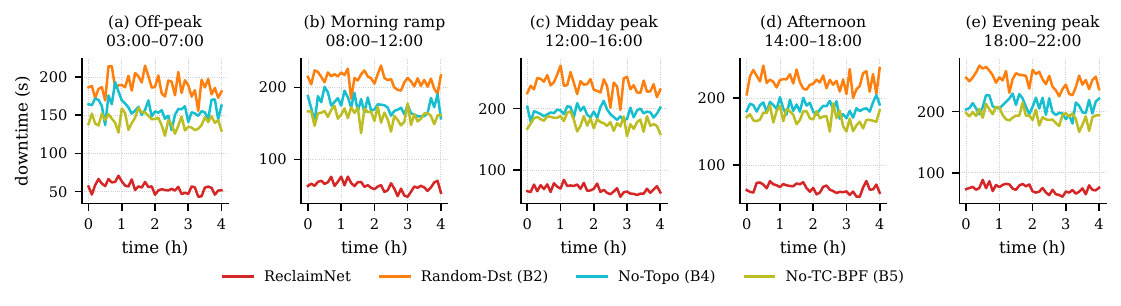}
    \caption{Per-migration downtime across five
      reclaim-intensity regimes drawn from the two-month
      trace: (a)~Off-peak (03:00--07:00),
      (b)~Morning ramp (08:00--12:00),
      (c)~Midday peak (12:00--16:00),
      (d)~Afternoon (14:00--18:00),
      (e)~Evening peak (18:00--22:00).
      \sys keeps median downtime within $\pm$10\% of
      off-peak even at evening peak.}
    \label{fig:downtime_regimes}
  \end{subfigure}

  \vspace{4pt}
  \begin{subfigure}[t]{\textwidth}
    \centering
    \includegraphics[width=\textwidth]{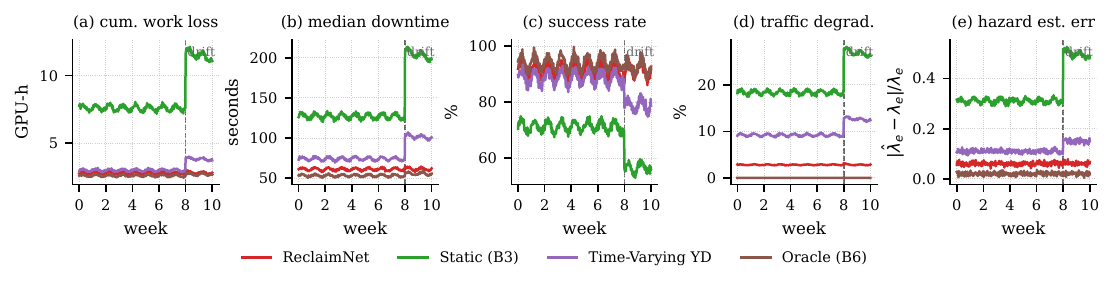}
    \caption{Multi-metric robustness across the two-month
      main deployment plus a held-out two-week drift period
      (dashed line marks drift onset):
      (a)~cumulative work loss, (b)~median downtime,
      (c)~success rate, (d)~research-traffic degradation,
      (e)~hazard-rate estimation error
      $|\hat\lambda_e-\lambda_e|/\lambda_e$.
      \sys recovers within $\sim$18\,min; baselines stay
      elevated.}
    \label{fig:drift_multimetric}
  \end{subfigure}

  \caption{End-to-end evaluation panels organised as
    multi-scenario sweeps. \emph{Top:}
    work-loss sensitivity to job class.
    \emph{Middle:} downtime sensitivity to reclaim intensity.
    \emph{Bottom:} multi-metric robustness to distribution
    drift. Shared method palette across all three rows.}
  \label{fig:eval_grid}
\end{figure*}

\subsection{Testbed and Baselines}

We evaluate on the 54-node testbed of
\S\ref{sec:measurement}, comparing \sys with seven online
baselines and an offline lower bound: \emph{Bamboo}~\cite{thorpe2023bamboo}
(B1), \emph{Slurm-PR} with DMTCP 10-min checkpoints
(B1$'$), \emph{No-Migration} (B2), trace-wide
\emph{Static-Ckpt} (B3, fixed periodic interval),
\emph{Random-Dst} (B4),
\emph{No-TC-BPF} (B5), state-independent
\emph{Time-Varying Fixed-Interval} (B7), and \emph{Oracle} (B6).
B4 and B5 remove topology-aware
destination selection and TC/eBPF isolation, respectively;
B7 keeps P1's estimators, the cube-root allocator
(Thm.~\ref{thm:p1coupled}), and P2/P3, but fixes one
$\Delta t$ per epoch; B6 solves an offline CPLEX ILP over
24-h windows. Full parameters and the ILP formulation appear
in Appendix~\ref{app:protocol-details}.

The two-month run comprises 312 GPU-training jobs
(BERT-Large, ResNet-152, GPT-2, Stable Diffusion), totaling
4{,}890 GPU-hours, and 847 departure events
(421 natural reclaims and 426 empirical injections shared by
all baselines). Natural-only results closely match aggregate
results---2.6 vs.\ 2.8\,GPU-h work loss, 58 vs.\ 61\,s median
downtime, 2.7\% vs.\ 2.8\% traffic degradation, and a 6.4\%
vs.\ 6.7\% B7 gap---with unchanged baseline ordering.

\begin{table}[t]
\centering
\caption{End-to-end evaluation results (two-month experiment,
847 departure events). $\downarrow$~lower is better;
$\uparrow$~higher is better.}
\label{tab:main_results}
\scriptsize
\setlength{\tabcolsep}{4pt}
\resizebox{\columnwidth}{!}{
\begin{tabular}{lccccc}
\toprule
System & Work & Down- & Mig. & Traffic & GPU \\
 & Loss\,$\downarrow$ & time (s)\,$\downarrow$ & Succ.\,\%\,$\uparrow$ & Degr.\,\%\,$\downarrow$ & Util.\,\%\,$\uparrow$ \\
 & (GPU-h) & (median) & & & \\
\midrule
Slurm-PR (B1$'$) & 8.2  & 195 & 64.5 & 26.4 & 38.1 \\
No-Mig. (B2)  & 18.4 & --- & 0    & 0.0  & 31.2 \\
Static (B3)   & 7.6  & 127 & 71.3 & 18.4 & 46.8 \\
Bamboo (B1)   & 4.5  & 75  & 85.4 & 8.2  & 58.7 \\
Rand-Dst (B4) & 3.8  & 98  & 84.2 & 7.3  & 59.1 \\
No-BPF (B5)   & 3.1  & 71  & 88.9 & 31.0 & 65.3 \\
TV-FI (B7)    & 3.0  & 64  & 90.4 & 2.9  & 67.5 \\
\midrule
\sys & 2.8 & 61 & 91.3 & 2.8 & 68.7 \\
Oracle (B6)   & 2.6  & 53  & 94.1 & 0.0  & 71.4 \\
\bottomrule
\end{tabular}
}
\end{table}

\noindent\sys reaches within 7.7\% of Oracle on work loss
while providing the only combination of $>91\%$ migration
success and $<3\%$ research traffic degradation; the
per-class, per-regime, and drift breakdowns in
Fig.~\ref{fig:eval_grid} (subfigures
\ref{fig:workloss_jobs}, \ref{fig:downtime_regimes},
\ref{fig:drift_multimetric}) confirm these aggregates are
not driven by a single workload or operating regime.

\subsection{E1: End-to-End Work Loss (Gap 1)}

\sys reduces average work loss by 66\% over
Slurm-PR (B1$'$) and 38\% over Bamboo (B1), while
remaining within 8\% of the Oracle lower bound
(B6). Its 6.7\% gain over Time-Varying Fixed-Interval (B7)
isolates the benefit of state-dependent interval selection
and is consistent with the $\sim$5.2\% adaptivity
gap predicted from $\Delta_{\mathrm{CS}}$
(Thm.~\ref{thm:p1corr}). The residual gap is explained by
diurnal non-stationarity: P1 re-estimates online, whereas B7
commits to an epoch-wide interval and consequently shows
higher downtime and lower migration success during peak
hours. Ablation (Table~\ref{tab:ablation}) attributes
64\% of the work-loss reduction to P1; P2+P3 recover
the remainder by preventing failed migrations. P1 shortens
checkpoint intervals to 4--6\,min under high reclaim risk and
relaxes them during low-risk periods to reduce network load.
All differences are significant on the 847-event population
(paired Wilcoxon, $p<0.01$); the two-month utilization trace
in Fig.~\ref{fig:e2e_utilization} visualises the headline
deployment gain.

\begin{figure}[t]
  \centering
  \includegraphics[width=\columnwidth]{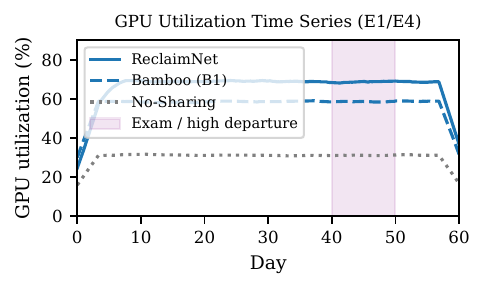}
  \caption{Two-month GPU utilization: \sys sustains $\sim$68.7\%, +10 pp over Bamboo and +37.5 pp over no-sharing.}
  \label{fig:e2e_utilization}
\end{figure}

\subsection{E2--E4: Downtime, Traffic Impact, Kill-Switch}

E2 Downtime (Gap 2). P2 reduces median migration
downtime by 38\% and p99 downtime by 52\%
relative to Random-Dst (B4) (Fig.~\ref{fig:p2_cdf}). The gain comes from preferring
same-building targets when feasible (71\% of migrations;
821\,Mbps intra-building vs.\ 243\,Mbps cross-building) and
falling back only when locality violates volatility or notice
constraints.
E3 Traffic impact (Gap 3). Under peak concurrent
migrations (up to seven events) with a co-running reference
workload, P3 holds research-traffic degradation to
2.8\%, compared with 31\% for No-TC-BPF (B5)
(Figs.~\ref{fig:p3_completion} and \ref{fig:p3_traffic}). Its
admission and rate enforcement reduce the unconstrained
burst from $\sim$6.8\,Gbps to $\leq$4.2\,Gbps, preserving
$\geq$30\% headroom.
E4 Kill-switch latency. TC BPF enforcement cuts
reclaim blocking latency to sub-millisecond scale (median
0.6\,ms, p99 0.9\,ms), over $100\times$
faster than application-layer kill paths; fewer than 500\,B
of new work can reach an evicting node after the signal.

\subsection{E5--E8: Scalability, Sensitivity, Ablation, Drift}

E5 Scalability. Mininet emulation to 500 nodes keeps
coordinator runtime below 18\,ms, far under the 10-second
notice budget (Appendix~\ref{app:eval-extra},
Fig.~\ref{fig:scalability}).
E6 Sensitivity. Sweeping
$\beta_{\text{p1}},\alpha,\beta_{\text{p3}}$ on the two-month
trace changes work loss/downtime by at most $\pm$6.4\% and
traffic degradation by $\pm$1.1\,pp; the defaults remain near
the empirical minima (Fig.~\ref{fig:sensitivity_grid}).
E7 Ablation. The per-protocol study
(Table~\ref{tab:ablation}) confirms orthogonal effects: P1
dominates work-loss reduction, P2 halves median downtime
(142\,s to 71\,s), and P3 reduces traffic degradation from
31.0\% to 2.8\%.
E8 Drift. Under a held-out two-week
deadline-crunch period with a $1.7\times$ hazard shift, P1
reconverges within 18\,min; work loss stays within
7\% of the in-distribution mean (95\% bootstrap CI
$[2.7,3.3]$ GPU-h, $p=0.21$).

\section{Discussion and Conclusion}
\label{sec:discussion}

\textbf{Limitations and scope.}
\sys relies on application-level checkpointing (PyTorch,
JAX hooks) and does not cover compiled or legacy GPU
workloads; CRIU-based CPU-state checkpointing with
GPU hooks is future work. The VVN model assumes trusted
providers; untrusted wide-area settings need additional
authentication and sandboxing.

\textbf{Partial deployment.}
Theorem~\ref{thm:p3isolation} holds only over the
\sys-controlled migration share: with participation
$\phi\in(0,1]$ and uncontrolled burst $\bar u$, the
guarantee degrades to $B - B_{\min} - (1-\phi)\bar u$.
Restoring it requires rate-capping uncontrolled tenants or
provisioning $B \geq B_{\min} + B_{\text{mig}} + (1-\phi)\bar u$.
Coordinator failure is handled by leader election with
at-most-once checkpoint-completion semantics.

\textbf{Conclusion.}
\sys couples reclaim-aware checkpointing (P1),
volatility/deadline-aware destination selection (P2), and
notice-driven migration admission with TC BPF enforcement
(P3). A 54-node campus testbed shows that provider autonomy
and system reliability coexist with lightweight coordination
and no switch reconfiguration. Future work targets joint
P1/P2 optimization, departure-forecast preemptive migration,
and RDMA checkpoint transfer.

\balance
\bibliographystyle{IEEEtran}
\bibliography{references}

\begin{thebibliography}{10}
\providecommand{\url}[1]{#1}
\csname url@samestyle\endcsname
\providecommand{\newblock}{\relax}
\providecommand{\bibinfo}[2]{#2}
\providecommand{\BIBentrySTDinterwordspacing}{\spaceskip=0pt\relax}
\providecommand{\BIBentryALTinterwordstretchfactor}{4}
\providecommand{\BIBentryALTinterwordspacing}{\spaceskip=\fontdimen2\font plus
\BIBentryALTinterwordstretchfactor\fontdimen3\font minus
  \fontdimen4\font\relax}
\providecommand{\BIBforeignlanguage}[2]{{%
\expandafter\ifx\csname l@#1\endcsname\relax
\typeout{** WARNING: IEEEtran.bst: No hyphenation pattern has been}%
\typeout{** loaded for the language `#1'. Using the pattern for}%
\typeout{** the default language instead.}%
\else
\language=\csname l@#1\endcsname
\fi
#2}}
\providecommand{\BIBdecl}{\relax}
\BIBdecl

\bibitem{kubernetes}
{Cloud Native Computing Foundation}, ``Kubernetes,''
  \url{https://kubernetes.io}, 2014, accessed: 2026-05-23.

\bibitem{yoo2003slurm}
A.~B. Yoo, M.~A. Jette, and M.~Grondona, ``{SLURM}: Simple linux utility for
  resource management,'' in \emph{Workshop on Job Scheduling Strategies for
  Parallel Processing}.\hskip 1em plus 0.5em minus 0.4em\relax Springer, 2003,
  pp. 44--60.

\bibitem{anderson2002seti}
D.~P. Anderson, J.~Cobb, E.~Korpela, M.~Lebofsky, and D.~Werthimer,
  ``{SETI@home}: An experiment in public-resource computing,''
  \emph{Communications of the ACM}, vol.~45, no.~11, pp. 56--61, 2002.

\bibitem{pande2003folding}
V.~S. Pande, I.~Baker, J.~Chapman \emph{et~al.}, ``Atomistic protein folding
  simulations on the submillisecond time scale using worldwide distributed
  computing,'' \emph{Biopolymers}, vol.~68, no.~1, pp. 91--109, 2003.

\bibitem{li2025gpunion}
Y.~Li, Y.~Zhang, H.~Liao, D.~Guo, and G.~Tang, ``{GPUnion}: Autonomous {GPU}
  sharing on campus,'' in \emph{Proceedings of the 24th ACM Workshop on Hot
  Topics in Networks (HotNets)}, 2025.

\bibitem{young1974}
J.~W. Young, ``A first order approximation to the optimum checkpoint
  interval,'' \emph{Communications of the ACM}, vol.~17, no.~9, pp. 530--531,
  1974.

\bibitem{daly2006}
J.~T. Daly, ``A higher order estimate of the optimum checkpoint interval for
  restart dumps,'' \emph{Future Generation Computer Systems}, vol.~22, no.~3,
  pp. 303--312, 2006.

\bibitem{thorpe2023bamboo}
J.~Thorpe, P.~Zhao, J.~Eyolfson, Y.~Qiao, Z.~Jia, M.~Zhang, R.~Netravali, and
  G.~H. Xu, ``Bamboo: Making preemptible instances resilient for affordable
  training of large {DNNs},'' in \emph{Proceedings of the 20th USENIX Symposium
  on Networked Systems Design and Implementation (NSDI '23)}.\hskip 1em plus
  0.5em minus 0.4em\relax USENIX Association, 2023, pp. 497--513.

\bibitem{xiao2018gandiva}
W.~Xiao, R.~Bhardwaj, R.~Ramjee, M.~Sivathanu, N.~Kwatra, Z.~Han, P.~Patel,
  X.~Peng, H.~Zhao, Q.~Zhang, F.~Yang, and L.~Zhou, ``Gandiva: Introspective
  cluster scheduling for deep learning,'' in \emph{Proceedings of the 13th
  USENIX Symposium on Operating Systems Design and Implementation (OSDI
  '18)}.\hskip 1em plus 0.5em minus 0.4em\relax USENIX Association, 2018, pp.
  595--610.

\bibitem{litzkow1988condor}
M.~J. Litzkow, M.~Livny, and M.~W. Mutka, ``Condor---a hunter of idle
  workstations,'' in \emph{Proceedings of the 8th International Conference on
  Distributed Computing Systems (ICDCS)}.\hskip 1em plus 0.5em minus
  0.4em\relax IEEE, 1988, pp. 104--111.

\bibitem{thain2005condor}
D.~Thain, T.~Tannenbaum, and M.~Livny, ``Distributed computing in practice: The
  {Condor} experience,'' \emph{Concurrency and Computation: Practice and
  Experience}, vol.~17, no. 2--4, pp. 323--356, 2005.

\bibitem{anderson2004boinc}
D.~P. Anderson, ``{BOINC}: A system for public-resource computing and
  storage,'' in \emph{Proceedings of the 5th IEEE/ACM International Workshop on
  Grid Computing (GRID '04)}.\hskip 1em plus 0.5em minus 0.4em\relax IEEE,
  2004, pp. 4--10.

\bibitem{bosilca2002mpichv}
G.~Bosilca, A.~Bouteiller, F.~Cappello, S.~Djilali, G.~Fedak, C.~Germain,
  T.~H{\'e}rault, P.~Lemarinier, O.~Lodygensky, F.~Magniette, V.~Neri, and
  A.~Selikhov, ``{MPICH-V}: Toward a scalable fault tolerant {MPI} for volatile
  nodes,'' in \emph{Proceedings of the ACM/IEEE Conference on Supercomputing
  (SC '02)}, 2002, pp. 29:1--29:18.

\bibitem{moody2010scr}
A.~Moody, G.~Bronevetsky, K.~Mohror, and B.~R. de~Supinski, ``Design, modeling,
  and evaluation of a scalable multi-level checkpointing system,'' in
  \emph{Proceedings of the ACM/IEEE International Conference for High
  Performance Computing, Networking, Storage and Analysis (SC '10)}.\hskip 1em
  plus 0.5em minus 0.4em\relax IEEE, 2010, pp. 1--11.

\bibitem{bautista2011fti}
L.~Bautista-Gomez, S.~Tsuboi, D.~Komatitsch, F.~Cappello, N.~Maruyama, and
  S.~Matsuoka, ``{FTI}: High performance fault tolerance interface for hybrid
  systems,'' in \emph{Proceedings of the ACM/IEEE International Conference for
  High Performance Computing, Networking, Storage and Analysis (SC '11)}, 2011,
  pp. 32:1--32:32.

\bibitem{plank1995libckpt}
J.~S. Plank, M.~Beck, G.~Kingsley, and K.~Li, ``Libckpt: Transparent
  checkpointing under {Unix},'' in \emph{Proceedings of the USENIX 1995
  Technical Conference}.\hskip 1em plus 0.5em minus 0.4em\relax USENIX
  Association, 1995, pp. 213--223.

\bibitem{hargrove2006blcr}
P.~H. Hargrove and J.~C. Duell, ``Berkeley lab checkpoint/restart ({BLCR}) for
  {Linux} clusters,'' \emph{Journal of Physics: Conference Series}, vol.~46,
  pp. 494--499, 2006.

\bibitem{ansel2009dmtcp}
J.~Ansel, K.~Arya, and G.~Cooperman, ``{DMTCP}: Transparent checkpointing for
  cluster computations and the desktop,'' in \emph{Proceedings of the 23rd IEEE
  International Parallel and Distributed Processing Symposium (IPDPS)}.\hskip
  1em plus 0.5em minus 0.4em\relax IEEE, 2009, pp. 1--12.

\bibitem{schroeder2007disk}
B.~Schroeder and G.~A. Gibson, ``Disk failures in the real world: What does an
  {MTTF} of 1,000,000 hours mean to you?'' in \emph{Proceedings of the 5th
  USENIX Conference on File and Storage Technologies (FAST '07)}.\hskip 1em
  plus 0.5em minus 0.4em\relax USENIX Association, 2007, pp. 1--16.

\bibitem{pinheiro2007disk}
E.~Pinheiro, W.-D. Weber, and L.~A. Barroso, ``Failure trends in a large disk
  drive population,'' in \emph{Proceedings of the 5th USENIX Conference on File
  and Storage Technologies (FAST '07)}.\hskip 1em plus 0.5em minus 0.4em\relax
  USENIX Association, 2007, pp. 17--28.

\bibitem{schroeder2009dram}
B.~Schroeder, E.~Pinheiro, and W.-D. Weber, ``{DRAM} errors in the wild: A
  large-scale field study,'' in \emph{Proceedings of the 11th International
  Joint Conference on Measurement and Modeling of Computer Systems (SIGMETRICS
  '09)}.\hskip 1em plus 0.5em minus 0.4em\relax ACM, 2009, pp. 193--204.

\bibitem{clark2005live}
C.~Clark, K.~Fraser, S.~Hand, J.~G. Hansen, E.~Jul, C.~Limpach, I.~Pratt, and
  A.~Warfield, ``Live migration of virtual machines,'' in \emph{Proceedings of
  the 2nd USENIX Symposium on Networked Systems Design and Implementation (NSDI
  '05)}.\hskip 1em plus 0.5em minus 0.4em\relax USENIX Association, 2005, pp.
  273--286.

\bibitem{hines2009postcopy}
M.~R. Hines, U.~Deshpande, and K.~Gopalan, ``Post-copy live migration of
  virtual machines,'' \emph{ACM SIGOPS Operating Systems Review}, vol.~43,
  no.~3, pp. 14--26, 2009.

\bibitem{wood2011cloudnet}
T.~Wood, K.~K. Ramakrishnan, P.~Shenoy, and J.~van~der Merwe, ``{CloudNet}:
  Dynamic pooling of cloud resources by live {WAN} migration of virtual
  machines,'' in \emph{Proceedings of the 7th ACM SIGPLAN/SIGOPS International
  Conference on Virtual Execution Environments (VEE '11)}.\hskip 1em plus 0.5em
  minus 0.4em\relax ACM, 2011, pp. 121--132.

\bibitem{gu2019tiresias}
J.~Gu, M.~Chowdhury, K.~G. Shin, Y.~Zhu, M.~Jeon, J.~Qian, H.~Liu, and C.~Guo,
  ``Tiresias: A {GPU} cluster manager for distributed deep learning,'' in
  \emph{Proceedings of the 16th USENIX Symposium on Networked Systems Design
  and Implementation (NSDI '19)}.\hskip 1em plus 0.5em minus 0.4em\relax USENIX
  Association, 2019, pp. 485--500.

\bibitem{mahajan2020themis}
K.~Mahajan, A.~Balasubramanian, A.~Singhvi, S.~Venkataraman, A.~Akella,
  A.~Phanishayee, and S.~Chawla, ``Themis: Fair and efficient {GPU} cluster
  scheduling,'' in \emph{Proceedings of the 17th USENIX Symposium on Networked
  Systems Design and Implementation (NSDI '20)}.\hskip 1em plus 0.5em minus
  0.4em\relax USENIX Association, 2020, pp. 289--304.

\bibitem{zhao2020hived}
H.~Zhao, Z.~Han, Z.~Yang, Q.~Zhang, F.~Yang, L.~Zhou, M.~Yang, F.~C.~M. Lau,
  Y.~Wang, Y.~Xiong, and B.~Wang, ``{HiveD}: Sharing a {GPU} cluster for deep
  learning with guarantees,'' in \emph{Proceedings of the 14th USENIX Symposium
  on Operating Systems Design and Implementation (OSDI '20)}.\hskip 1em plus
  0.5em minus 0.4em\relax USENIX Association, 2020, pp. 515--532.

\bibitem{peng2018optimus}
Y.~Peng, Y.~Bao, Y.~Chen, C.~Wu, and C.~Guo, ``Optimus: An efficient dynamic
  resource scheduler for deep learning clusters,'' in \emph{Proceedings of the
  13th European Conference on Computer Systems (EuroSys '18)}.\hskip 1em plus
  0.5em minus 0.4em\relax ACM, 2018, pp. 3:1--3:14.

\bibitem{qiao2021pollux}
A.~Qiao, S.~K. Choe, S.~J. Subramanya, W.~Neiswanger, Q.~Ho, H.~Zhang, G.~R.
  Ganger, and E.~P. Xing, ``Pollux: Co-adaptive cluster scheduling for
  goodput-optimized deep learning,'' in \emph{Proceedings of the 15th USENIX
  Symposium on Operating Systems Design and Implementation (OSDI '21)}.\hskip
  1em plus 0.5em minus 0.4em\relax USENIX Association, 2021, pp. 1--18.

\bibitem{xiao2020antman}
W.~Xiao, S.~Ren, Y.~Li, Y.~Zhang, P.~Hou, Z.~Li, Y.~Feng, W.~Lin, and Y.~Jia,
  ``{AntMan}: Dynamic scaling on {GPU} clusters for deep learning,'' in
  \emph{Proceedings of the 14th USENIX Symposium on Operating Systems Design
  and Implementation (OSDI '20)}.\hskip 1em plus 0.5em minus 0.4em\relax USENIX
  Association, 2020, pp. 533--548.

\bibitem{jeon2019analysis}
M.~Jeon, S.~Venkataraman, A.~Phanishayee, J.~Qian, W.~Xiao, and F.~Yang,
  ``Analysis of large-scale multi-tenant {GPU} clusters for {DNN} training
  workloads,'' in \emph{Proceedings of the 2019 USENIX Annual Technical
  Conference (ATC '19)}.\hskip 1em plus 0.5em minus 0.4em\relax USENIX
  Association, 2019, pp. 947--960.

\bibitem{narayanan2019pipedream}
D.~Narayanan, A.~Harlap, A.~Phanishayee, V.~Seshadri, N.~R. Devanur, G.~R.
  Ganger, P.~B. Gibbons, and M.~Zaharia, ``{PipeDream}: Generalized pipeline
  parallelism for {DNN} training,'' in \emph{Proceedings of the 27th ACM
  Symposium on Operating Systems Principles (SOSP '19)}.\hskip 1em plus 0.5em
  minus 0.4em\relax ACM, 2019, pp. 1--15.

\bibitem{huang2019gpipe}
Y.~Huang, Y.~Cheng, A.~Bapna, O.~Firat, D.~Chen, M.~Chen, H.~Lee, J.~Ngiam,
  Q.~V. Le, Y.~Wu, and Z.~Chen, ``{GPipe}: Efficient training of giant neural
  networks using pipeline parallelism,'' in \emph{Advances in Neural
  Information Processing Systems 32 (NeurIPS)}, 2019, pp. 103--112.

\bibitem{li2020pytorchddp}
S.~Li, Y.~Zhao, R.~Varma, O.~Salpekar, P.~Noordhuis, T.~Li, A.~Paszke,
  J.~Smith, B.~Vaughan, P.~Damania, and S.~Chintala, ``{PyTorch} distributed:
  Experiences on accelerating data parallel training,'' \emph{Proceedings of
  the VLDB Endowment}, vol.~13, no.~12, pp. 3005--3018, 2020.

\bibitem{wang2023topoopt}
W.~Wang, M.~Khazraee, Z.~Zhong, M.~Ghobadi, Z.~Jia, D.~Mudigere, Y.~Zhang, and
  A.~Kewitsch, ``{TopoOpt}: Co-optimizing network topology and parallelization
  strategy for distributed training jobs,'' in \emph{Proceedings of the 20th
  USENIX Symposium on Networked Systems Design and Implementation (NSDI
  '23)}.\hskip 1em plus 0.5em minus 0.4em\relax USENIX Association, 2023, pp.
  739--767.

\bibitem{sharma2015spotcheck}
P.~Sharma, S.~Lee, T.~Guo, D.~Irwin, and P.~Shenoy, ``{SpotCheck}: Designing a
  derivative {IaaS} cloud on the spot market,'' in \emph{Proceedings of the
  10th European Conference on Computer Systems (EuroSys '15)}.\hskip 1em plus
  0.5em minus 0.4em\relax ACM, 2015, pp. 16:1--16:15.

\bibitem{harlap2017proteus}
A.~Harlap, A.~Tumanov, A.~Chung, G.~R. Ganger, and P.~B. Gibbons, ``Proteus:
  Agile {ML} elasticity through tiered reliability in dynamic resource
  markets,'' in \emph{Proceedings of the 12th European Conference on Computer
  Systems (EuroSys '17)}.\hskip 1em plus 0.5em minus 0.4em\relax ACM, 2017, pp.
  589--604.

\bibitem{hindman2011mesos}
B.~Hindman, A.~Konwinski, M.~Zaharia, A.~Ghodsi, A.~D. Joseph, R.~H. Katz,
  S.~Shenker, and I.~Stoica, ``Mesos: A platform for fine-grained resource
  sharing in the data center,'' in \emph{Proceedings of the 8th USENIX
  Symposium on Networked Systems Design and Implementation (NSDI '11)}.\hskip
  1em plus 0.5em minus 0.4em\relax USENIX Association, 2011, pp. 295--308.

\bibitem{vavilapalli2013yarn}
V.~K. Vavilapalli, A.~C. Murthy, C.~Douglas, S.~Agarwal, M.~Konar, R.~Evans,
  T.~Graves, J.~Lowe, H.~Shah, S.~Seth, B.~Saha, C.~Curino, O.~O'Malley,
  S.~Radia, B.~Reed, and E.~Baldeschwieler, ``{Apache Hadoop YARN}: Yet another
  resource negotiator,'' in \emph{Proceedings of the 4th Annual Symposium on
  Cloud Computing (SoCC '13)}.\hskip 1em plus 0.5em minus 0.4em\relax ACM,
  2013, pp. 5:1--5:16.

\bibitem{verma2015borg}
A.~Verma, L.~Pedrosa, M.~Korupolu, D.~Oppenheimer, E.~Tune, and J.~Wilkes,
  ``Large-scale cluster management at {Google} with {Borg},'' in
  \emph{Proceedings of the 10th European Conference on Computer Systems
  (EuroSys '15)}.\hskip 1em plus 0.5em minus 0.4em\relax ACM, 2015, pp.
  18:1--18:17.

\bibitem{mccanne1993bpf}
S.~McCanne and V.~Jacobson, ``The {BSD} packet filter: A new architecture for
  user-level packet capture,'' in \emph{Proceedings of the USENIX Winter 1993
  Conference}.\hskip 1em plus 0.5em minus 0.4em\relax USENIX Association, 1993,
  pp. 259--269.

\bibitem{hoiland2018xdp}
T.~H{\o}iland-J{\o}rgensen, J.~D. Brouer, D.~Borkmann, J.~Fastabend,
  T.~Herbert, D.~Ahern, and D.~Miller, ``The {eXpress Data Path}: Fast
  programmable packet processing in the operating system kernel,'' in
  \emph{Proceedings of the 14th International Conference on Emerging Networking
  EXperiments and Technologies (CoNEXT '18)}.\hskip 1em plus 0.5em minus
  0.4em\relax ACM, 2018, pp. 54--66.

\bibitem{jouet2017bpfabric}
S.~Jouet and D.~P. Pezaros, ``{BPFabric}: Data plane programmability for
  software defined networks,'' in \emph{Proceedings of the ACM/IEEE Symposium
  on Architectures for Networking and Communications Systems (ANCS '17)}.\hskip
  1em plus 0.5em minus 0.4em\relax IEEE, 2017, pp. 38--48.

\bibitem{saeed2017carousel}
A.~Saeed, N.~Dukkipati, V.~Valancius, V.~T. Lam, C.~Contavalli, and A.~Vahdat,
  ``Carousel: Scalable traffic shaping at end hosts,'' in \emph{Proceedings of
  the ACM SIGCOMM 2017 Conference}.\hskip 1em plus 0.5em minus 0.4em\relax ACM,
  2017, pp. 404--417.

\bibitem{alizadeh2010dctcp}
M.~Alizadeh, A.~Greenberg, D.~A. Maltz, J.~Padhye, P.~Patel, B.~Prabhakar,
  S.~Sengupta, and M.~Sridharan, ``Data center {TCP} ({DCTCP}),'' in
  \emph{Proceedings of the ACM SIGCOMM 2010 Conference}.\hskip 1em plus 0.5em
  minus 0.4em\relax ACM, 2010, pp. 63--74.

\bibitem{alizadeh2013pfabric}
M.~Alizadeh, S.~Yang, M.~Sharif, S.~Katti, N.~McKeown, B.~Prabhakar, and
  S.~Shenker, ``{pFabric}: Minimal near-optimal datacenter transport,'' in
  \emph{Proceedings of the ACM SIGCOMM 2013 Conference}.\hskip 1em plus 0.5em
  minus 0.4em\relax ACM, 2013, pp. 435--446.

\bibitem{alizadeh2014conga}
M.~Alizadeh, T.~Edsall, S.~Dharmapurikar, R.~Vaidyanathan, K.~Chu,
  A.~Fingerhut, V.~T. Lam, F.~Matus, R.~Pan, N.~Yadav, and G.~Varghese,
  ``{CONGA}: Distributed congestion-aware load balancing for datacenters,'' in
  \emph{Proceedings of the ACM SIGCOMM 2014 Conference}.\hskip 1em plus 0.5em
  minus 0.4em\relax ACM, 2014, pp. 503--514.

\bibitem{bai2015pias}
W.~Bai, L.~Chen, K.~Chen, D.~Han, C.~Tian, and H.~Wang, ``Information-agnostic
  flow scheduling for commodity data centers,'' in \emph{Proceedings of the
  12th USENIX Symposium on Networked Systems Design and Implementation (NSDI
  '15)}.\hskip 1em plus 0.5em minus 0.4em\relax USENIX Association, 2015, pp.
  455--468.

\bibitem{chowdhury2014varys}
M.~Chowdhury, Y.~Zhong, and I.~Stoica, ``Efficient coflow scheduling with
  {Varys},'' in \emph{Proceedings of the ACM SIGCOMM 2014 Conference}.\hskip
  1em plus 0.5em minus 0.4em\relax ACM, 2014, pp. 443--454.

\end{thebibliography}

\appendices

\section{Proofs of Theorems and Lemmas}
\label{app:proofs}

\subsection{Proof of Lemma~\ref{lem:p1local} (Local Optimal Interval)}
Setting $\partial \mathcal{L}/\partial(\Delta t) = 0$:
$\tfrac{\lambda_e}{2} - \tfrac{C}{B_{\text{eff}}(\Delta t)^2}
= 0$, giving $\Delta t^* = \sqrt{2C / (\lambda_e
B_{\text{eff}})}$. The second derivative
$\tfrac{2C}{B_{\text{eff}}(\Delta t)^3} > 0$ confirms a
global minimum.

\subsection{Proof of Theorem~\ref{thm:p1coupled} (Network-Coupled Allocation)}
For fixed $b_i$, minimizing
\eqref{eq:p1multi-cost} over $\Delta_i$ gives
$\Delta_i^*(b_i)=\sqrt{2C_i/(\lambda_i b_i)}$ by
Lemma~\ref{lem:p1local}, and the resulting cost is
$\mathcal{L}_i^*(b_i)=\sqrt{2\lambda_iC_i/b_i}$.
The remaining problem is convex in $b_i>0$:
$\min_{\{b_i\}}\sum_i\sqrt{2\lambda_iC_i}\,b_i^{-1/2}$
s.t.\ $\sum_i b_i\leq B_{\mathrm{ckpt}}$.
The bandwidth constraint is tight at optimum. The KKT
stationarity condition is
$-\tfrac{1}{2}\sqrt{2\lambda_iC_i}\,b_i^{-3/2}+\mu=0$,
so $b_i\propto(\lambda_iC_i)^{1/3}$. Normalizing by
$\sum_i b_i=B_{\mathrm{ckpt}}$ gives \eqref{eq:p1cube};
substitution gives \eqref{eq:p1multi-delta}.

\subsection{Proof of Lemma~\ref{lem:p1box} (Box-Constrained Allocation)}
The objective is convex and separable, and the constraint
set is a polytope. KKT stationarity together with
complementary slackness gives: if a box constraint is
inactive at the optimum, the cube-root rule holds among the
remaining jobs sharing the residual budget; if a box
constraint is active, the corresponding $b_i=W_i$.
At every iteration, the active set strictly shrinks: a job
is only removed if its unconstrained allocation under the
current residual budget exceeds its cap, in which case it
is fixed at the cap. Since the residual budget is
monotonically non-increasing and the number of jobs is
finite, the iteration terminates in at most $K$ rounds.
If the procedure halted with some cap still violated, the
residual problem would be the unconstrained one whose
solution satisfies all remaining caps by construction, a
contradiction. The returned point therefore satisfies KKT
and, by strict convexity of the objective in each $b_i>0$,
is the unique global minimizer.

\subsection{Proof of Theorem~\ref{thm:p1corr} (Adaptivity Gap)}
The state-dependent optimum at $(\lambda_e,B)$ is
$\mathcal{L}^{*}(\lambda_e,B)=\sqrt{2C\lambda_e/B}
=\sqrt{2C}\cdot W$ (Lemma~\ref{lem:p1local}), so
$\mathbb{E}[\mathcal{L}_{\text{dep}}^{*}]=\sqrt{2C}\,
\mathbb{E}[W]$. For $\pi_{\text{ind}}$, since the cost rate
under a fixed $\Delta t_0$ is linear in $\lambda_e$ and
$1/B_{\text{eff}}$, the expected cost rate is
$\mathbb{E}[\mathcal{L}(\Delta t_0)]=
(\Delta t_0/2)\bar\lambda+(C/\Delta t_0)\bar\Theta$;
minimizing over $\Delta t_0$ yields
$\Delta t_0^{*}=\sqrt{2C\bar\Theta/\bar\lambda}$ and
$\mathbb{E}[\mathcal{L}_{\text{ind}}^{*}]=
\sqrt{\,2C\,\bar\lambda\,\bar\Theta\,}$ (the square root
covers the entire product).
Applying the $L^{2}(\mathbb{P})$ Cauchy--Schwarz inequality
to $X=\sqrt{\lambda_e}$ and $Y=\sqrt{1/B_{\text{eff}}}$ gives
$(\mathbb{E}[XY])^{2}=(\mathbb{E}[W])^{2}\leq
\mathbb{E}[X^{2}]\,\mathbb{E}[Y^{2}]=\bar\lambda\bar\Theta$,
which establishes \eqref{eq:gap-exact} and \eqref{eq:cs}.
Equation \eqref{eq:gap-quant} follows from
$a-b=(a^{2}-b^{2})/(a+b)$ with
$a=\sqrt{\bar\lambda\bar\Theta}$ and $b=\mathbb{E}[W]$.
Equality holds iff $Y=cX$ a.s., i.e.\ $\lambda_e B_{\text{eff}}$
is a.s.\ constant.

\subsection{Proof of Lemma~\ref{thm:p2} (Hierarchical Locality)}
For any same-building destination $d_l$,
$T_{\text{mig}}(s,d_l,t)\leq C/B_{\text{local}}^{-}+T_r$.
For any cross-building destination $d_x$,
$T_{\text{mig}}(s,d_x,t)\geq C/B_{\text{core}}^{+}+T_r$.
The bandwidth separation assumption implies the former is
no larger than the latter. The full P2 score still compares
survival risk among local candidates and falls back to the
global candidate set when too few local nodes are feasible.

\subsection{Proof of Theorem~\ref{thm:p3completion} (Migration Completion)}
Transfer time at rate $r_k$ is
$C_k/r_k \leq \tau_k - T_{r,k}$. Adding restart time
$T_{r,k}$ gives total migration time $\leq \tau_k$.

\subsection{Proof of Theorem~\ref{thm:p3isolation} (Traffic Isolation)}
P3 admits and allocates migration flows only within
$B_{\text{mig}}=B-B_{\text{res}}\leq B-B_{\min}$, and the
TC BPF token buckets enforce the assigned rates at the
controlled egress points. Thus the aggregate migration load
on the protected bottleneck is no greater than
$B-B_{\min}$, leaving at least $B_{\min}$ residual capacity
for non-migration traffic.

\section{Extended Evaluation: Scalability, Sensitivity, Ablation, Drift}
\label{app:eval-extra}

\subsection{Scalability via Emulation (E5)}
Using Mininet with 54 real measurements as ground truth,
we emulate up to 500 nodes. Coordinator scheduling
latency grows sub-linearly (from \textbf{3\,ms} at
54 nodes to \textbf{18\,ms} at 500 nodes), below the
10-second minimum notice period in all tested scenarios.
Algorithm~\ref{alg:p2}'s $O(|\mathcal{D}_{\text{cand}}|
\cdot L)$ complexity yields near-linear scaling;
Algorithm~\ref{alg:p1}'s per-node computation adds
$<0.1$\,ms per node (Fig.~\ref{fig:scalability}).

\begin{figure}[t]
  \centering
  \includegraphics[width=\columnwidth]{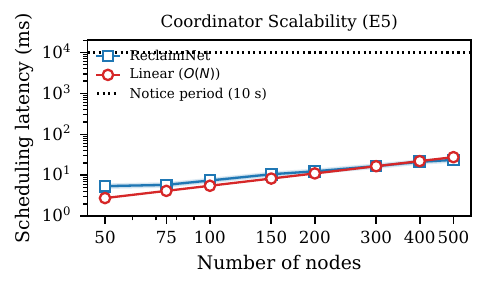}
  \caption{Coordinator scheduling latency vs.\ cluster
    size (emulation). \sys scales sub-linearly from
    3\,ms (54 nodes) to 18\,ms (500 nodes), well below the
    minimum 10-second provider notice period.}
  \label{fig:scalability}
\end{figure}

\subsection{Hyperparameter Sensitivity (E6)}
We sweep the three principal hyperparameters
($\beta_{\text{p1}}\in\{0.1,0.2,0.3\}$ for P1's
network-feasibility floor, $\alpha\in\{0.5,1.0,2.0\}$ for
P2's stability penalty, $\beta_{\text{p3}}\in\{0.2,0.3,0.4\}$
for P3's research-traffic reservation) on the same two-month
trace (Fig.~\ref{fig:sensitivity_grid}).
Defaults were calibrated on weeks 1--2 (held out from the
evaluation in Tables~\ref{tab:main_results}--\ref{tab:ablation},
which use weeks 3--8); the swept points cover $\pm$50\% to
$\pm$100\% around each default. All three parameters operate
in stable regions: median work loss varies by at most
\textbf{$\pm$5.7\%} across the swept ranges, downtime by
\textbf{$\pm$6.4\%}, and research-traffic degradation by
\textbf{$\pm$1.1 pp}. The default values
($\beta_{\text{p1}}=0.2$, $\alpha=1.0$, $\beta_{\text{p3}}=0.3$)
sit near the empirical minimum of each metric.

\begin{figure}[t]
  \centering
  \includegraphics[width=\columnwidth]{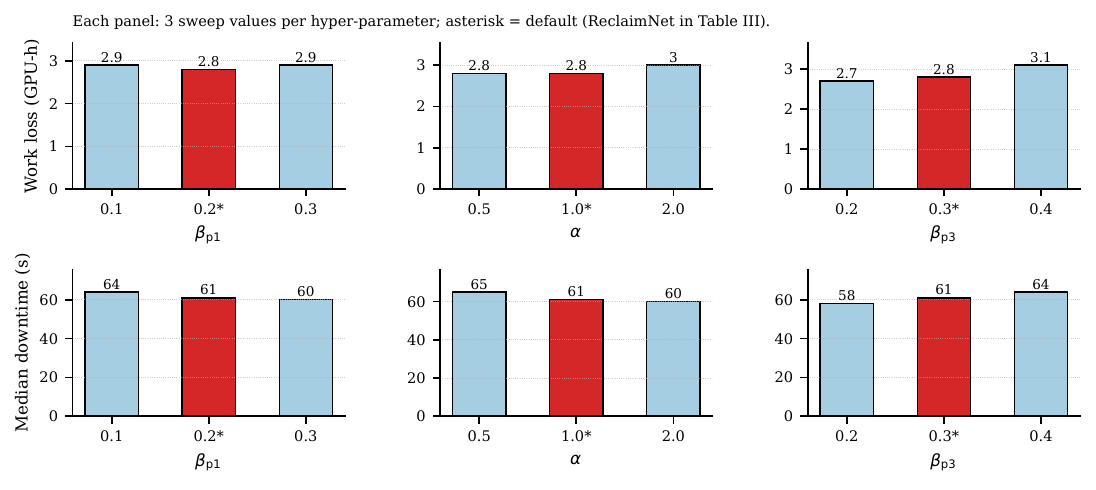}
  \caption{Hyperparameter sensitivity rendered as a 2$\times$3
    sweep grid (rows: work loss / median downtime; columns:
    $\beta_{\text{p1}}$, $\alpha$, $\beta_{\text{p3}}$). The
    default (asterisk, red bar) reproduces \sys in
    Table~\ref{tab:main_results}; all six panels remain flat
    within $\pm$6.4\%.}
  \label{fig:sensitivity_grid}
\end{figure}

\subsection{Per-Protocol Ablation (E7)}
We run all seven $\{$P1,P2,P3$\}$ subsets against the same
trace (Table~\ref{tab:ablation}). The all-off row uses a
deliberately naive baseline---fixed 30-min checkpoint
interval, uniformly random destination, no traffic shaping.
This row is strictly worse than Bamboo (B1) because Bamboo's
pipeline-redundancy recovery re-executes only affected
micro-batches, while the all-off configuration accepts
larger rollbacks. The three two-protocol rows exactly match
baselines B5 (P1+P2 = No-TC-BPF), B4 (P1+P3 = Random-Dst),
and a P2+P3-only configuration, confirming that
Tables~\ref{tab:main_results} and~\ref{tab:ablation} are
derived from the same runs.

\begin{table}[t]
\centering
\caption{Per-protocol contribution. The all-off row is a naive
baseline; rows P1+P2, P1+P3, P1+P2+P3 reproduce baselines
B5, B4, and \sys from Table~\ref{tab:main_results}.}
\label{tab:ablation}
\scriptsize
\setlength{\tabcolsep}{4pt}
\begin{tabular}{cccccc}
\toprule
\textbf{P1} & \textbf{P2} & \textbf{P3} & \textbf{Loss} & \textbf{Down.} & \textbf{Degr.} \\
            &             &             & (GPU-h) & (s) & (\%) \\
\midrule
\ding{55} & \ding{55} & \ding{55} & 8.7 & 145 & 34.8 \\
\ding{51} & \ding{55} & \ding{55} & 4.9 & 142 & 33.6 \\
\ding{55} & \ding{51} & \ding{55} & 8.1 &  79 & 32.4 \\
\ding{55} & \ding{55} & \ding{51} & 8.5 & 144 &  3.1 \\
\ding{51} & \ding{51} & \ding{55} & 3.1 &  71 & 31.0 \\
\ding{51} & \ding{55} & \ding{51} & 3.8 &  98 &  7.3 \\
\ding{55} & \ding{51} & \ding{51} & 4.6 &  67 &  2.9 \\
\textbf{\ding{51}} & \textbf{\ding{51}} & \textbf{\ding{51}} & \textbf{2.8} & \textbf{61} & \textbf{2.8} \\
\bottomrule
\end{tabular}
\end{table}

\subsection{Robustness to Distribution Drift (E8)}
We re-evaluate \sys on a held-out two-week period spanning
the end-of-semester deadline crunch (beyond the main
two-month deployment), during which the empirical hazard
rate $\lambda_e$ shifts upward by $1.7\times$ relative to
the training window. P1's sliding-window MLE re-converges
within \textbf{18\,min} of the regime change (criterion
$|\hat\lambda_e-\lambda_e^{\mathrm{true}}|/\lambda_e^{\mathrm{true}}<0.10$
sustained for two windows); end-to-end work loss on the
held-out period is \textbf{3.0\,GPU-h} (95\% bootstrap CI
$[2.7, 3.3]$ over 1{,}000 resamples; within 7\% of the
in-distribution number, $p=0.21$), confirming that the
adaptive controller does not require trace-specific
calibration.

\section{Protocol Derivation Details}
\label{app:protocol-details}

\subsection{Adaptivity-Gap Binding Regimes}

Because $B_{\text{eff}}=\min(W,B(t))$, in epochs where the
local write speed $W$ is the binding term, $B_{\text{eff}}$
collapses to a constant and the corresponding contribution
to $\Delta_{\mathrm{CS}}$ vanishes. On our two-month trace,
intra-building NVMe sustains $W\!\approx\!1.1$\,GB/s while
$B(t)$ averages 820\,Mbps intra-building and 240\,Mbps
across core; $B(t)\!\ll\!W$ in nearly all loaded epochs, so
$B$ is the binding term in $\sim$97\% of 10-minute epochs.
Conditioning the gap estimator on $B$-binding epochs alone
yields a predicted gap of $\sim$5.4\%; the deployment-wide
5.2\% figure is mixed with the $W$-binding minority.
Deployments with faster access links (e.g.\ 25\,GbE) shift
more epochs into the $W$-binding regime and see a
\emph{smaller} gap; those with slower SSDs at fixed access
bandwidth see the gap remain close to the $B$-binding value.

\subsection{Calibration Procedure for the Gap}

We estimate the three sufficient statistics from the
calibration trace (weeks 1--2): the sample means
$\hat{\bar\lambda}$, $\hat{\bar\Theta}$, and the sample mean
$\widehat{\mathbb{E}[W]}$ of
$\sqrt{\lambda_e/B_{\text{eff}}}$ over the same epochs.
Plugging into \eqref{eq:gap-exact} yields a predicted ratio
of $\sqrt{\hat{\bar\lambda}\hat{\bar\Theta}}/
\widehat{\mathbb{E}[W]}\approx 1.052$, i.e.\ a $\sim$5.2\%
adaptivity gap on this trace. This ratio strictly above one
is the direct evidence that $\lambda_e B_{\text{eff}}$ is
not a.s.\ constant on the deployment trace; the
anti-correlation and CV figures in \S\ref{sec:measurement} are
consistent with that gap but are not, by themselves, the
equality condition.

\subsection{Full Baseline Parameter Settings}

We fix all baseline hyperparameters from their original papers
or published configurations.
\emph{Static-Ckpt (B3)} is parameterized analytically as
$\bar C_{\text{time}}=\bar C_{\text{size}}/\bar B
=(3.2\,\text{GB}\!\cdot\!8)/(0.530\,\text{Gbps})\approx 48.3$\,s
with the trace-wide hazard $\bar\lambda_e\approx 0.99/\text{h}$
(358 emergency events over 361.4 spot-priced GPU-h), giving
$\Delta t^{*}\!\approx\!9.8$\,min; an empirical grid search
over $\Delta t\!\in\!\{5,10,15\}$\,min selects 10\,min
independently (7.6\,GPU-h, 71.3\% success), which we adopt.
Alternative denominators (total GPU-h $\to$ 36\,min;
natural-only emergencies $\to$ 13.6\,min) yield strictly
worse loss and are reported in our artefact.
\emph{Bamboo (B1)} reuses its public pipeline-redundancy
implementation with provider departure mapped to a preemption
event. \emph{Slurm-PR (B1$'$)} runs SLURM with
\texttt{--requeue} and a DMTCP 10-min checkpoint.
\emph{Random-Dst (B4)} reuses P1+P3 with destinations sampled
uniformly among VRAM-compatible nodes.
\emph{No-TC-BPF (B5)} disables P3 admission and edge
enforcement; flows use vanilla TCP CUBIC.
\emph{Time-Varying Fixed-Interval (B7)} re-estimates
$\bar\lambda_e,\bar\Theta=\overline{1/B}$ each epoch with the
same eBPF estimators and the cube-root allocator
(Thm.~\ref{thm:p1coupled}), but commits to
$\Delta t=\sqrt{2C_{\text{size}}\bar\Theta/\bar\lambda_e}$
per epoch---the $\pi_{\text{ind}}$ policy of
Thm.~\ref{thm:p1corr}; P2 and P3 are unchanged so the
comparison isolates the adaptivity-gap contribution.

\subsection{Oracle (B6) ILP Formulation}

(B6) is solved offline given the full two-month reclaim
and bandwidth trace. With $x_{j,t,d}\in\{0,1\}$ (job $j$
checkpointed at time $t$ to destination $d$):
$\min \sum_{j}\!\sum_{e_j}\Delta_{e_j}^{\mathrm{rb}}(x)$
s.t.\ $\sum_d x_{j,t,d}\!\leq\!1\ \forall j,t$,
$\sum_j r_j x_{j,t,d}\!\leq\! w(e,t)\ \forall e\in E,t$,
$T_{\mathrm{mig}}(j,t,d)\!\leq\!\tau_{j,t}$ for scheduled
events. We solve in 24-h batches with CPLEX 22.1.1 (mean
wall-clock 4.7\,h per window on a 64-core server);
infeasible online but a tight upper bound.

\section{Implementation Details}
\label{app:impl}

\paragraph{Software stack.} \sys is implemented in
$\sim$5{,}400 LoC. \emph{Node agents} ($\sim$1{,}500 LoC
modern C++17, built with CMake and gRPC) manage container
lifecycle, trigger checkpoints through PyTorch
\texttt{torch.save} and Hugging Face
\texttt{save\_pretrained} bindings invoked via a thin
\texttt{pybind11} bridge, invoke P2's destination selection,
and update BPF maps through \texttt{libbpf}. \emph{eBPF
programs} ($\sim$800 LoC C, Clang 17 + \texttt{libbpf}) implement
P3's TC BPF token-bucket and kill-switch using CO-RE so a
single object loads across the heterogeneous campus kernels.
The \emph{coordinator} ($\sim$2{,}000 LoC C++17) runs the P3
admission/enforcement loop, the P2 topology database
(refreshed by periodic eBPF probes), and P1's departure-rate
estimators plus cube-root allocator; persistent state lives
in a local SQLite store and Prometheus metrics are exposed on
\texttt{/metrics}. Containers run under Docker with the NVIDIA
Container Toolkit; GPU isolation uses MIG on A100/A6000 and
process-level VRAM reservation on RTX consumer cards. Node
agents drive \texttt{tc}-BPF attachment and map updates
directly via \texttt{libbpf}, so loading and reconfiguration
require no node reboots.

\paragraph{Security considerations.} Inter-node communication
is authenticated via mTLS (TLS 1.3) with short-lived JWTs
issued at registration. Containers run under Docker Seccomp
profiles restricted to CUDA and \texttt{libbpf} syscalls; MIG
hardware partitioning protects A100/A6000 GPUs, while RTX
3090/4090 nodes zero VRAM pages on container teardown before
returning GPUs to providers. The kill-switch TC BPF program
is a protocol contribution (\S\ref{sec:p3}) distinct from
these engineering safeguards.

\end{document}